\documentclass{ws-ijfcs}
\usepackage{enumerate}
\usepackage{cite}
\usepackage{url}
\usepackage{xspace}
\usepackage{epsfig}
\usepackage{graphicx}
\usepackage{pgfplots}
\usetikzlibrary{matrix}
\usetikzlibrary{positioning,shapes,arrows}
\usepackage{color}
\usepackage{float}
\newcommand{\const}{d}
\newcommand{\perfratio}{performance ratio}

\newcommand{\alg}{{\ensuremath{\mathbb{A}}}\xspace}
\newcommand{\algprime}{{\ensuremath{\mathbb{A}'}}\xspace}
\newcommand{\algprimeprime}{{\ensuremath{\mathbb{A}''}}\xspace}
\newcommand{\opt}{\ensuremath{\text{OPT}}\xspace}
\newcommand{\OPT}{\ensuremath{\text{OPT}}\xspace}
\newcommand{\cratio}[1]{\ensuremath{\text{CR}_{#1}}\xspace}
\newcommand{\Litem}{\text{item of size $L$}\xspace}
\newcommand{\Sitem}{\text{item of size $S$}\xspace}
\newcommand{\Ls}{items of size $L$\xspace}
\newcommand{\Ss}{items of size $S$\xspace}
\newcommand{\SC}{\ensuremath{\text{countS}}\xspace}
\newcommand{\LC}{\ensuremath{\text{countL}}\xspace}
\newcommand{\CS}{\ensuremath{\text{countS}^\prime}\xspace}
\newcommand{\CL}{\ensuremath{\text{countL}^\prime}\xspace}
\newcommand{\SB}{\ensuremath{\text{S-bins}}\xspace}
\newcommand{\LB}{\ensuremath{\text{L-bins}}\xspace}
\newcommand{\Spr}{\ensuremath{\text{S-bins}^\prime}\xspace}
\newcommand{\Lpr}{\ensuremath{\text{L-bins}^\prime}\xspace}
\newcommand{\SBt}{\ensuremath{\text{late-S-bins}}\xspace}
\newcommand{\SBtpr}{\ensuremath{\text{late-S-bins}^\prime}\xspace}
\newcommand{\LBt}{\ensuremath{\text{late-L-bins}}\xspace}
\newcommand{\LBtpr}{\ensuremath{\text{late-L-bins}^\prime}\xspace}
\newcommand{\newalg}{\ensuremath{\text{2-Phase-Packer}}\xspace}
\newcommand{\maxbinsize}{\ensuremath{M}\xspace}
\newcommand{\tab}{\hspace*{6mm}}
\newcommand{\nats}{{\mathbb N}}

\newtheorem{observation}{Observation}

\begin{document}
\markboth{Joan Boyar, Faith Ellen}
{Tight Bounds for Restricted Grid Scheduling}

\title{Tight Bounds for Restricted Grid Scheduling\footnote{A preliminary 
version
of this paper, entitled ``Bounds for Scheduling Jobs on Grid Processors'',
appeared in {\em  Space-Efficient Data Structures, Streams, 
and Algorithms}, Lecture
Notes in Computer Science, volume 8066, 2013, pages 12-26.}}

\author{Joan Boyar}

\address{Department of Mathematics and Computer Science,
University of Southern Denmark,\\
Campusvej 55, DK-5230 Odense~M, Denmark\\
\email{joan@imada.sdu.dk}
}

\author{Faith Ellen}
\address{Department of Computer Science, University of Toronto,\\
10 King's College Road, Toronto, Ontario,
Canada M5S~3G4\\
\email{faith@cs.toronto.edu}
}

\maketitle

\begin{history}
\received{(Day Month Year)}
\accepted{(Day Month Year)}
\comby{(xxxxxxxxxx)}
\end{history}

\begin{abstract}
The following
problem is considered:
Items with integer sizes are given and variable sized
bins arrive online. A bin must be used if there is still
an item remaining which fits in it when the bin arrives.
The goal is to minimize the total size of all the bins used.
Previously, a lower bound of $\frac{5}{4}$
on the competitive ratio of this problem was achieved using
items of size $S$ and $2S-1$.
For these
item sizes
and maximum bin size $M=4S-3$, 
we obtain asymptotically matching upper
and lower bounds, which vary depending on the ratio of the
number of small items to the number of large items.
\end{abstract}

\keywords{
Online algorithms; variable-sized bin packing; Restricted Grid Scheduling.}

\section{Introduction}
In the classical online bin packing problem, bins of unit size
are given and items of varying size, each at most $M$,  arrive online. This is reversed
in the Grid Scheduling
problem~\cite{BF10}, where items and their sizes are given
and bins of varying size arrive online. In this problem,
a bin must be used if there is at least one unpacked item that fits in it.
As in the classical problem, the items must all be packed in the bins, so that 
the total size of the items packed in a bin is no more than the size of the bin.
The goal is to minimize the sum of the sizes of bins used to pack the items:
Bins which were unused because no unpacked items fit in them do not 
contribute to this total.
It is assumed that there are plenty of bins that are large enough to hold the 
largest items, so that feasibility is not an issue.
 
The Grid Scheduling problem is known to have competitive ratio at least $\frac{5}{4}$~\cite{BF10}.
This was achieved using items of size $S$ and $L = 2S-1$,
for integers $S\geq 2$. 
In this paper, we consider the Restricted Grid Scheduling problem,
where only these item sizes are allowed and the maximum possible bin size
is $M=4S-3$.
For this restricted problem, we obtain asymptotically matching upper and lower bounds, which depend
on the ratio of $s$, the number of small  items, to $\ell$, the number of large
items. Specifically, if $s$ items of size $S > 1$ and $\ell$ items of size $L=2S-1$ are given
and the maximum bin size is $M = 4S-3$, the competitive ratio for this problem is
$$
\begin{array}{ll}
1 + \frac{1}{2 + s/\ell}
& \mbox{ if }  \ell \leq s/2,\\

1 + \frac{1}{4}
& \mbox{ if }  s/2 < \ell \leq 5s/6, \mbox{ and}\\

1 + \frac{2}{3 + 6\ell/s}
& \mbox{ if }  5s/6< \ell.
\end{array}
$$
This bound is plotted in Figure \ref{graph}.

\begin{figure}[t]
\begin{center}
\begin{tikzpicture}[scale=3]
\draw [<->] (0,1.3) node [above] {\begin{tabular}{c}competitive \\ ratio\end{tabular}} -- (0,0) node [below] {0} -- (2.0,0)
node [right] {$\ell/s$};
\draw[very thick,domain=0:0.5] plot (\x, {1.0+\x/(2*\x+1)});
\draw[thin] (-0.01,1.25) node [left] {1.25} -- (0.01,1.25);
\draw[thin] (-0.01,1.00) node [left] {1.00} -- (0.01,1.00);
\draw[thick,dashed] (0.5,1.25) -- (0.5,0.0) node [below] {1/2};
\draw[very thick,domain=0.5:0.8333] plot (\x, {1.25});
\draw[thick,dashed] (0.8333,1.25) -- (0.8333,0.0) node [below] {5/6};
\draw[very thick,domain=0.8333:2.0] plot (\x, {1.0+2/(3+6*\x)});
\end{tikzpicture}
\end{center}
\caption{The competitive ratio of the Restricted Grid Scheduling  problem as a function of $\ell/s$}
\label{graph}
\end{figure}
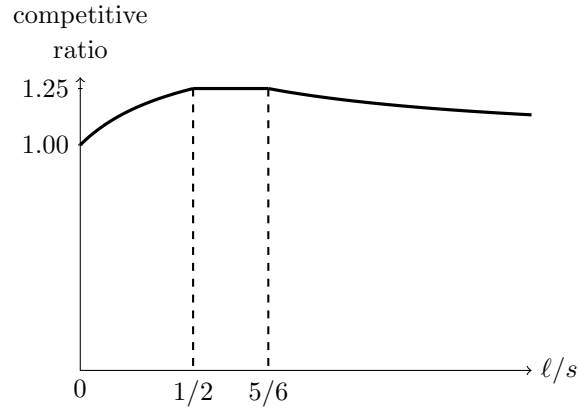

{\bf The Grid Scheduling Problem}
The Grid Scheduling problem~\cite{BF10,BF11,EFL11}
was proposed in an article concerning a Grid architecture~\cite{Vinter05}.
The proposed application in that article concerned the
well-known bioinformatics program BLAST. More generally, one considers
a large problem that
is divided up into independent jobs, which may be of
varying sizes.
These jobs are distributed to heterogeneous processors  (for example, personal
computers) having idle time.
These processors become available in an on-line manner, as they become idle. 
Folding@Home (in which idle processors are used to find new ways to fold proteins)
and SETI@home (in which idle processors are used to look for signs of extraterrestrial life)
are two other examples.
Such systems have been called
Grids,
denoting computation distributed seamlessly over (possibly) large distances,
in the same way that electricity is distributed over the electrical grid.
It does not refer to computation in which the network topology
is a rectangular grid.
Note that other researchers have used the terminology ``Grid scheduling''
to denote other problems which are very different from the one considered here.

It is well-known that paging may slow down computation drastically. In fact,
parallelizing jobs can result in
superlinear speed-up by eliminating unnecessary paging~\cite{DCF03}.
Paging can be avoided by assigning jobs to a processor whose memory capacity is at least
the combined memory requirements of those jobs.
In the Grid Scheduling problem,
the item sizes are the memory requirements of the jobs and the bin
sizes are the memory capacity of the processors. 
To finish the jobs as quickly as possible, any processor that arrives and can process
at least one of the remaining jobs should be assigned to it (rather than waiting for a processor
of close to the right size, which might never come).
For the Grid Scheduling problem,
this condition is equivalent to 
requiring that any arriving bin should be used, if there is an unpacked item that fits into it.
We prove, in fact, that items should be packed into an arriving bin until there are no more
items that can fit into it.
Allowing bins which are smaller than the maximum item size reflects
the situation where some processors might not have enough memory to handle
the largest jobs.

{\bf Related Work}
In~\cite{BF10}, Boyar and Favrholdt
presented an algorithm to solve the Grid Scheduling problem with
competitive ratio $\frac{13}{7}$. 
They also proved that the competitive ratio 
of any algorithm for the Grid Scheduling problem
is at least $\frac{5}{4}$.
In that proof, the adversary used only
two different item sizes, $ S$ and $L=2S-1$, and four different bin sizes,
$ S$, $2S-1$, $2S$, and $4S-4 $, for an integer $S > 1$.

This $\frac{5}{4}$ lower bound bound has not yet been improved, despite
its
distance from the $\frac{13}{7}$ upper bound. Here, we contribute to 
understanding why improving the lower bound may be difficult:
We show
that the lower bound
is tight in the sense that there is an algorithm with a competitive
ratio of $\frac{5}{4}$ when there are only items of size $S$ and $L=2S-1$
and the maximum bin size is $4S-3$. In addition to allowing all possible
bin sizes in this range, we consider all possible ratios of
the number of items of size $S$ to the number of items of size $L$ and show
that the competitive ratio is $\frac{5}{4}$ in a middle range, but decreases
as the ratio moves away from this middle range.
The bounds obtained are tight. Some of the results concerning packings
are general, not restricted to this special case, and might be useful in
narrowing the gap between the upper and lower bounds for the general
problem.

Grid Scheduling with Conflicts,
the Grid Scheduling problem with the restriction that certain pairs of 
jobs cannot be executed on the
same processor (i.e.~be placed in the same bin), was later studied by 
Epstein et al.\ in~\cite{EFL11}.
In the semi-online version, where the bins come in non-increasing
order, they show a 
lower bound of $\frac{5}{4}$ and an upper bound of $3$ on the competitive ratio.

In Zhang's Bin Packing problem~\cite{Z97},
it is also the case that  items are given in advance and bins arrive online.
Bins and items can 
have any sizes in the range $(0,1]$,
but the smallest bin is at least as large as the largest item,
so that each item can be packed in any bin. Zhang showed that
analogues of four classical bin packing algorithms, including
First-Fit Decreasing, all have a competitive ratio of $2$.
Algorithms for the Grid Scheduling problem also apply to Zhang's Bin Packing problem
with the same competitive ratio.
In particular, the algorithm in~\cite{BF10},
with a competitive ratio of $\frac{13}{7}$,  solves
an open question proposed in~\cite{Z97}, asking if there exists an
algorithm with competitive ratio less than $2$.
However,
lower bounds for the Grid Scheduling Problem do not apply to Zhang's Bin Packing problem,
since the Grid Scheduling problem does not
restrict the sizes of the bins to be at least as large as every item.

In many applications of bin packing, there are only a small number
of different item sizes. A number of papers have
considered the problem of packing a sequence of items
of  two different sizes in bins of size 1 in an online manner.
In particular, there is a lower bound of 4/3 on the competitive ratio
\cite{L80,GJY06} and a matching upper bound \cite{GJY06}.
When both item sizes are bounded above by $1/k$, 
the competitive ratio can be improved to $\frac{(k+1)^2}{k^2 +k+1}$
\cite{EL08}.

{\bf Overview}
We begin with some preliminaries,
including
a formal definition
of the Grid Scheduling problem.
In Section~\ref{sec:lower}, we
prove
the lower
bound on the competitive ratio of the Restricted Grid Scheduling 
problem (and, hence, of the Grid Scheduling problem).
Then we provide a matching upper bound
for the Restricted Grid Scheduling problem, under the assumption that
the maximum possible bin size, $M$, is at most $4S-3$.
As discussed in Section~\ref{prelim}, one has to assume that $M$ is
bounded (as compared to $S$)
in the general Grid Scheduling problem 
for any algorithm to be competitive. Thus, our algorithm
 shows that any better lower bound
for the general Grid Scheduling problem must use more than two item sizes,
 different item sizes, or
a larger maximum bin size.
In Section~\ref{sec:properties}, we
present
properties of optimal packings
for the Restricted Grid Scheduling problem.
Some of these provide motivation for the design of our algorithm,
while others are important for the analysis.
We show why a few simple
algorithms are not optimal in Section \ref{sec:simple}.
Our algorithm, \newalg,
appears
in Section~\ref{sec:upper}, together with the analysis.
In Section~\ref{sec:relax}, we consider a slightly relaxed version
of the Grid Scheduling problem and show that it is equivalent to the
version we defined.
We conclude with some open questions.

\section{Preliminaries}
\label{prelim}

Given a set of items, each with a positive integer size,
and a sequence of bins,
each with a positive integer size at most $M$,
the goal
of the Grid Scheduling problem
is to pack all the items in the bins so that the sum of
the sizes of the items packed in each bin is at most the size of the bin
and the sum of the sizes of bins used is minimized.
The bins in the sequence arrive one at a time and each must be packed
before the next bin arrives, without knowledge of the sizes
of any future bins.
If a bin is at least as large as the smallest unpacked item, it
must be packed with at least one item.
Moreover, after it has been packed, no unpacked item fits into the space that remains in the bin.
There is no cost for a bin that is smaller than 
the smallest unpacked item.
Note, there is no loss of generality
in assuming that every item has size at most $M$,
because no item of size greater than $M$ can be packed.

Because all items are required to be packed,
it is assumed that enough sufficiently large bins arrive.
Thus, any algorithm eventually packs all items.
For example, it suffices that every sequence 
has a suffix consisting of bins of size $M$, whose length is equal to the number of items.

The {\em cost} of a packing is the sum of the sizes of the bins it uses
(i.e. in which it places at least one item).
A packing is {\em optimal}
if it is feasible and
its cost
is not more than that of any other feasible packing
for the same set of items and sequence of bins.

The competitive ratio \cite{BEY98,KMRS88,ST85} of an on-line algorithm is 
the worst-case ratio of
the on-line performance to the optimal off-line performance, up to an
additive constant.
More precisely, for a set $I$ of items (or, equivalently, a multi-set of item sizes), 
a sequence $\sigma$ of bins, and an
algorithm \alg for the Grid Scheduling problem, let $\alg(I,\sigma)$ denote
the cost of the packing produced 
by \alg when packing $I$ in the sequence $\sigma$ of bins. Then,
the {\em competitive ratio} $\cratio{\alg}$ of \alg is 
$$\cratio{\alg} =
\inf
\left\{ c \mid \exists \const ,\forall I, \forall \sigma,
\alg(I,\sigma) \leq c \cdot \opt(I,\sigma) + \const \right\},$$
where $\opt(I, \sigma)$ denotes the minimum cost of any feasible packing of $I$ in the sequence $\sigma$ of bins 
i.e. produced by an optimal off-line algorithm.
For specific choices of
families of increasingly large sets $I_n$ and sequences $\sigma_n$,
with $n \in \nats$,
the performance ratios, 
$\frac{\alg(I_n,\sigma_n)}{\opt(I_n,\sigma_n)}$, can be used to prove a lower bound on the
competitive ratio of $\alg$.

If there is no bound on the maximum bin size,
the Grid Scheduling problem is uninteresting, because the competitive ratio is unbounded:
Once enough bins for an optimal packing have arrived,
an adversary could give bins of arbitrarily large size, which
the algorithm would be forced to use.
This can be seen in the
following example: Suppose there are two items of size $S$ and one of size $L = 2S-1$.
The first bin has
size $2S$. If the algorithm puts at least one of the small items there, 
the adversary
next gives two bins of size $S$, followed by one of size
$M$.
If the
algorithm puts the larger item in the first bin, the adversary next 
gives one bin of
size $L$, followed by one of size $M$.
The ratio of the algorithm's performance compared to the adversary's
is at least $\min\{(2S+M)/4S, 
(2S+L+M)/(2S+L)\}$.
This is close to
$\frac{3}{2}$ when $M=4S-3$,
but  it is unbounded if $M$ can be arbitrarily large compared to $S$.
Note that this example also shows
a lower bound of $\frac{3}{2}$ on the {\em strict}
competitive ratio (the competitive ratio where the additive constant
in the definition is zero), even when $M=4S-3$,
but not on the competitive ratio.

The Restricted Grid Scheduling problem is the restriction of the Grid Scheduling
problem 
in which all items have size $S$ or $L=2S-1$, where $M = 4S-3$ and $S > 1$.
Note that the maximum size $M=4S-3$ was chosen to be as large as possible
without allowing space for two items of size $L$.
There is no loss of generality
in assuming that every bin has size at least $S$,
because no bin of size less than $S$ can be used to pack an item.
The value $S$ is assumed to be a constant, to allow an algorithm
to use one extra bin of size $M$, for example, which
is only counted in the additive constant of the competitive ratio.
The definition of competitive ratio allows the additive constant to
be arbitrarily large, as long as it is
independent of the number of items.

\section{Lower Bounds}
\label{sec:lower}

In this section, we assume that there can be an unbounded number of items of size $S$ and
an unbounded number of  items of size $L$. Otherwise, an algorithm that
gives lower preference to the items of 
which there are only a bounded number, 
has 
competitive ratio $1$,
taking the additive constant to be $\const = \min\{s,\ell\} \cdot M$.

\begin{theorem}
No algorithm for 
the Restricted Grid Scheduling problem
has competitive ratio lower than
$$\left \{ \begin{array}{ll}
1 + \frac{1}{2 + s/\ell}
& \mbox{ if }  \ell \leq s/2,\\

1 + \frac{1}{4}
& \mbox{ if }  s/2 < \ell \leq 5s/6, \mbox{ and}\\

1 + \frac{2}{3 + 6\ell/s}
& \mbox{ if }  5s/6< \ell.
\end{array} \right .$$
\label{thm:lower}
\end{theorem}

\begin{proof}
Consider an algorithm for the Restricted Grid Scheduling problem and
an instance in which 
there are $s$ items of size $S>1$,
$\ell$ items of size $L=2S-1$, and maximum bin size 
$M=4S-3$.
We start with the case when $\ell \leq s/2$ and then handle the 
case when $\ell > s/2$. In both cases, we consider two
subcases, depending on how the algorithm packs the first
batch of bins.

\medskip

\noindent
{\bf Case I}: $\ell \leq s/2$.

The adversary begins by giving $\ell$
bins of size $2S$.
In each of these bins, the algorithm must pack either
two items of size $S$ or one item of size $L$.
Let $0 \leq k \leq \ell$ be the number of these bins
in which the algorithm packs two items of size $S$.
Then the algorithm has $s-2k$ items of size $S$
and $k$ items of size $L$ left to pack.

\medskip

{\bf Case I.1}: $k \leq \ell/2$.\\
Next,
the adversary gives
$s-2\ell$ bins of size $S$, followed by $2\ell$ bins of size $L$.
The algorithm must pack one item of size $S$ in each bin of size $S$
and must use one bin of size $L$ for each of the remaining
$s-2k -(s-2\ell) = 2(\ell-k)$ items of size $S$
and $k$ items of size $L$.
The total cost incurred by the algorithm is
$\ell \cdot 2S + (s-2\ell)\cdot S + (2\ell - k ) \cdot L
= s \cdot S + 2 \ell \cdot L - k \cdot L
\geq s \cdot S + 3\ell \cdot L/2$.

For this sequence, OPT packs two items of size $S$
in each of the $\ell$ bins of size $2S$, one item of size $S$ in each of
the $s-2\ell$ bins of size $S$, and 
one item of size $L$ in each of the next $\ell$ bins of size $L$, for
total cost $s\cdot S + \ell \cdot L$.
Thus, the \perfratio{}  of the algorithm
is at least
$$\frac{s \cdot S + 3\ell \cdot L/2}{s\cdot S + \ell \cdot L}
 =  1 + \frac{\ell \cdot L}{2\ell \cdot L + 2s \cdot S }
 =  1 + \frac{1}{2+\frac{s}{\ell} + \frac{s/\ell}{2 S-1}}
\rightarrow 1 + \frac{1}{2+ s/\ell} \mbox{ as } S \rightarrow \infty .$$

{\bf Case I.2}: $k > \ell/2$.\\
Next, the adversary  gives $s$ bins of size $S$, followed by
$\ell$ bins of size $M$.
The algorithm packs one item of size $S$ in the
first $s-2k \geq 2\ell - 2k \geq 0$
of these bins, using up all its items of size $S$.
It discards the remaining $2k$ bins of size $S$,
because it has no remaining elements that are small enough
to fit in them. Then  the algorithm packs its remaining $k$ items
of size $L$ into $k$ bins of size $M = 4S -3$.
The total cost incurred by the algorithm is
$$\ell \cdot 2S + (s-2k) \cdot S + k \cdot (4S-3)
= (2\ell + s ) \cdot S + k \cdot (2S - 3)
> (2\ell + s ) \cdot S + \ell \cdot (S - 3/2).$$

For this sequence,  OPT packs one item of size $L$
in each of the $\ell$ bins of size $2S$ and
one item of size $S$ in each of the next $s$ bins.
The total cost of OPT's packing is $\ell \cdot 2S + s \cdot S $.
Thus, the
\perfratio{} of the algorithm is
greater than
$$\frac{(3\ell + s) \cdot S - 3\ell/2}{(2\ell + s) \cdot S }
= 1 + \frac{1 - 3/(2S)}{2 + s/\ell }
\rightarrow 1 + \frac{1}{2 + s/\ell} \mbox{ as } S \rightarrow \infty .
$$

\noindent
{\bf Case II}: $\ell > s/2$.

The adversary begins by giving $\lfloor s/2 \rfloor$
bins of size $2S$.
In each of these bins, the algorithm must pack either
two items of size $S$ or one item of size $L$.
Let $0 \leq k \leq \lfloor s/2 \rfloor$ be the number of these bins
in which the algorithm packs two items of size $S$.
Then the algorithm has $s-2k$ items of size $S$
and $\ell - \lfloor s/2 \rfloor + k$ items of size $L$ left to pack.

{\bf Case II.1}: $k \leq \lfloor s/ 2 \rfloor - s/8 - \ell/4 + 1$
or $k \leq \lfloor s/ 2 \rfloor - s/3 + 1$.\\
Next, the adversary gives
$\lceil s/2 \rceil - \lfloor s/2 \rfloor$ bins of size $S$
(i.e. one bin of size $S$
if $s$ is odd and no bins of size $S$ if $s$ is even),
$\lfloor s/2 \rfloor -k + \ell -1$ bins of size $L$,
and one bin of size $M$.
Since $(s-2k) + 
(\ell - \lfloor s/2 \rfloor + k) =  (\lceil s/2\rceil -\lfloor s/2 \rfloor)  + (\lfloor s/2 \rfloor -k + \ell -1) +1$,
the algorithm packs one of its remaining items in each 
of these bins,
so the total cost it incurs 
is
$ \lfloor s/2 \rfloor \cdot 2S + (\lceil s/2 \rceil -
\lfloor s/2 \rfloor) \cdot S +
(\lfloor s/2 \rfloor - k+ \ell -1) \cdot L + M
 =  s \cdot S + \ell \cdot L + 
(\lfloor s/2 \rfloor -k+1) \cdot L -1$.

We may assume, without loss of generality, that $s \geq 4$. Then
$\lfloor s/2 \rfloor -k \geq
\min\{ s/8 + \ell/4, s/3\} -1 > s/4 - 1 \geq 0$,
so there are at least $\ell$ bins of size $L$.
For this sequence, OPT packs two items of size $S$
in each of the $\lfloor s/2 \rfloor$ bins of size $2S$, one item of size $S$
in the bin of size $S$, if $s$ is odd,
and one item of size $L$ in each of the next $\ell$ bins of size $L$.
Its total cost is $s\cdot S + \ell \cdot L$.

If $k \leq \lfloor s/ 2 \rfloor - s/8 - \ell/4 + 1$,
then $(\lfloor s/ 2 \rfloor -k+1) \cdot L \geq (s/8 + \ell/4) \cdot L
=  S \cdot s/4 - s/8 + L \cdot \ell/4$,
so the
\perfratio{} of the algorithm
is at least
$$\frac{(s \cdot S + \ell \cdot L ) \cdot 5/4 -s/8 -1}
{s\cdot S + \ell \cdot L}
= \frac{5}{4} - \frac{s/8 + 1}{s\cdot S + \ell \cdot (2S-1)}
\rightarrow \frac{5}{4}
\mbox{   as } S \rightarrow \infty.$$
Similarly, if $k \leq \lfloor s/ 2 \rfloor - s/3 + 1$,
then $(\lfloor s/ 2 \rfloor -k+1) \cdot L \geq (s/3) \cdot L$,
so the
\perfratio{}  of the algorithm
is at least
$$1 + \frac{(s/3) \cdot L}{s\cdot S + \ell \cdot L} 
= 1 + \frac{2 - 1/S}{3+6\ell/s -3\ell/(s \cdot S)}
\rightarrow 1 +\frac{2}{3+6\ell/s}
\mbox{   as } S \rightarrow \infty.$$
{\bf Case II.2}: $k > \lfloor s/ 2 \rfloor - s/8 - \ell/4 + 1$
and
$k > \lfloor s/ 2 \rfloor - s/3 + 1 > s/6$.\\
Next,
the adversary gives
$\max \{ s - 2k, s -\ell + \lfloor s/2\rfloor\}$ bins of size $S$,
followed by $\min\{ 2k, \ell-\lfloor s/2\rfloor\}$ bins of size $L+S$,
$\max\{\ell - \lfloor s/2\rfloor- 2k ,0\}$  bins of size $L$, and finally
$k$ bins of size $M=4S-3$.
The algorithm packs its $s-2k$ items of size $S$ into
bins of size $S$.
Since $\min\{ 2k, \ell-\lfloor s/2\rfloor\} +
\max\{\ell - \lfloor s/2\rfloor- 2k ,0\} +k = 
\ell -  \lfloor s/2\rfloor + k$, the algorithm
packs one item of size $L$ in each bin
of size $L+S$, $L$, and $M$.
The total cost incurred by the algorithm  is
$$\begin{array}{l}
\lfloor s/2\rfloor \cdot 2S + (s-2k) \cdot S +
\min\{ 2k, \ell-\lfloor s/2\rfloor\} \cdot (L+S)  \\
~~~~~~ + \max\{\ell - \lfloor s/2\rfloor- 2k ,0\} \cdot L 
+ k \cdot M\\
= (2 \lfloor s/2\rfloor + s) \cdot S 
+ (\ell - \lfloor s/2\rfloor) \cdot L + k \cdot (2S-3) +
\min\{2k,\ell - \lfloor s/2\rfloor\} \cdot S.
\end{array}$$
For this sequence, \opt\ fills every bin it uses except for
the $\lfloor s/2\rfloor$ bins of size $2S$, in which it puts items
of size $L = 2S-1$. Therefore, the total cost of \opt's packing
is $s\cdot S + \ell \cdot L +  \lfloor s/2\rfloor$.

If $2k \geq \ell - \lfloor s/2\rfloor$, the
\perfratio{} of the algorithm
is at least
\begin{eqnarray*}
&& \frac{(2\lfloor s/2\rfloor+s) \cdot S
+ (\ell-\lfloor s/2\rfloor )\cdot L
+ k \cdot (2S-3) +
(\ell - \lfloor s/2\rfloor) \cdot S}
{s \cdot S +\ell \cdot L +\lfloor s/2\rfloor} \\
&& > \frac{(\ell+s + \lfloor s/2\rfloor)\cdot S + (\ell - \lfloor s/2\rfloor)
\cdot L + (\lfloor s/2\rfloor - s/8 - \ell/4 + 1) \cdot (2S-3)}
{s \cdot S +\ell \cdot L +\lfloor s/2\rfloor} \\
&& = \frac{5}{4}
+ \frac{(\lfloor s/2\rfloor - s/2 +2) \cdot S +
\ell  +3s/8 - 13\lfloor s/2 \rfloor/4 - 3}
{(s +2\ell) \cdot S - \ell +\lfloor s/2\rfloor} \\
&& >\frac{5}{4} + \frac{\ell  +3s/8 - 13\lfloor s/2 \rfloor/4 - 3}
{(s +2\ell) \cdot S - \ell +\lfloor s/2\rfloor} 
\rightarrow \frac{5}{4}
\mbox{   as } S \rightarrow \infty.
\end{eqnarray*}

If $2k \leq \ell - \lfloor s/2\rfloor$, the
\perfratio{} of the algorithm
is at least
\begin{eqnarray*}
& &\frac{(2\lfloor s/2\rfloor+s) \cdot S + (\ell-\lfloor s/2\rfloor )\cdot L
+ k \cdot (4S-3)}{s \cdot S +\ell \cdot L +\lfloor s/2\rfloor} \\
& &= 1 + \frac{k \cdot (4S-3)}{(s + 2\ell) \cdot S - \ell  +\lfloor s/2\rfloor} \\
& &> 1 + \frac{(s/6) \cdot (4S-3)}{(s+ 2\ell) \cdot S +
\lfloor s/2\rfloor - \ell} \\
& &= 1 + \frac{2- 3/2S}{3 + 6\ell/s + 3(\lfloor s/2\rfloor - \ell)/sS}
\rightarrow 1 + \frac{2}{3+6\ell/s} \mbox{ as } S \rightarrow \infty.
\end{eqnarray*}

Note that $\frac{5}{4} \leq 1 + \frac{2}{3+6\ell/s}$ if and only if
$\ell \leq 5s/6$. Thus, $\frac{5}{4}$ is a lower bound on the
competitive ratio 
when $s/2 < \ell \leq 5s/6$
and $1 + \frac{2}{3+6\ell/s}$ is a lower bound on the
competitive ratio when $\ell > 5s/6$.
\end{proof}

The above lower bound uses limits taken as the size, $S$, approaches
infinity, despite the value $S$ being considered a constant. The point
is that, for any constant $\epsilon$, there exists a value for $S$ making 
the lower bound greater than the actual ratio minus $\epsilon$.
A similar proof gives the same ratio if one normalizes the sizes
(to
$M=1$, $S=1/4+\epsilon$ and $L=1/2+\delta$, where
$0<\delta<2\epsilon$) and lets  $\epsilon$ approach zero.
It is simply necessary that $L$ be asymptotically close to, but less than $2S$, and
$M$ be asymptotically close to, but less than $2L$.

\section{Properties of Optimal Packings for Restricted Grid Scheduling}
\label{sec:properties}

From now on, we assume that there are no bins of size less than $S$.
Intuitively, it seems
bad
to place an item of size $S$ in
a bin if the algorithm could have placed an item of size $L$ there
instead.
We say that a bin used in a packing is {\em bad}
if it contains
at least one
item of size $S$,
it has empty space at least $L-S$, and
some later bin contains an item of size $L$.
Note that a bin containing an item of size $L$
and an item of size $S$ has empty space at most
$M-L-S = L-S-1$, so it is not bad.

\begin{lemma}
For any finite set of items and any sequence of bins,
there exists an optimal 
packing that contains no bad bin.
\label{emptyspace}
\end{lemma}

\begin{proof}
Let $p$ be any optimal 
packing of a finite set of items $I$ into a sequence of bins 
$\sigma = \langle b_1, b_2, \ldots, b_m\rangle$.
Assume the claim is true for any smaller set of items.

Suppose $p$ contains a bad bin. Let $b_f$ be the first bad bin in $p$ and let $b_\ell$ be the last bin in $p$ that contains an item of size $L$. Then, by definition of bad, $f < \ell$
and, by definition of 
the Grid Scheduling problem, the empty space in $b_f$ is less than $L$.

First suppose that $p$ has an empty bin between $b_f$ and $b_\ell$.
Let $b_k$ be the first such bin.
Since $p$ is feasible, $size(b_k) < L$ and only items of size $L$ are packed in bins $b_k,\ldots,b_m$.
At most one item of size $L$ can be packed in any bin (of size at most $M$), so
each nonempty bin after $b_k$,
including $b_\ell$, contains exactly one item, which is of size $L$.
Consider the packing $p'$ obtained from $p$
by moving one item of size $S$ from bin $b_f$ to
bin $b_k$ and moving the item of size $L$ in bin $b_\ell$ to bin $b_f$.
Then $b_\ell$ is empty in the packing $p'$.
Since $p$ is feasible and 
each bin in $p'$
(except for the unused bin, $b_{\ell}$,)
is as full as the
corresponding bin in $p$,
it follows that
$p'$ is also feasible.
But the cost of $p'$ is equal to  $cost(p) - size(b_\ell)
+ size(b_k) < cost(p)$, since $size(b_k) < L \leq size(b_\ell)$.
This contradicts the optimality of $p$.

Therefore $p$ has no empty bins between $b_f$ and $b_\ell$.
Then a  packing $p'$ can be obtained from $p$ by switching an item of size
$S$ in $b_f$ with the item of size $L$ in $b_\ell$.
Note that $p'$ is optimal, since $p$ is. 
Since $b_f$ contains an item of size $L$, it is not bad in $p'$.
Since $p$ is 
feasible,
no item packed in bin $b_f$ or later fits in any bin prior to $b_f$.
The same is true for $p'$, since $p$ and $p'$ are the same prior to bin $b_f$.
Furthermore,
in $p'$, no item packed in a bin after $b_f$ will fit in bin $b_f$,
since its empty space, which was less than $L$ in $p$, is
less than $L + S - L  = S$ in $p'$.

Let $J \subsetneq I$ be the set of items that $p'$ packs into $\sigma'  = \langle b_{f+1},  \ldots, b_m\rangle$.
By the induction hypothesis, there is an optimal 
packing $q$ of $J$ into $\sigma'$ that contains
no bad bins.
Let $p''$ be the packing of $I$ into $\sigma$ that packs each item in $I-J$
into the same bin that $p'$ does and packs each item in $J$ into the same bin that $q$ does.
Then $p''$ is an optimal 
packing without bad bins.
By induction, the claim is true for all finite sets of items, $I$.
\end{proof}

Lemma~\ref{emptyspace} motivates the following definition.

\begin{definition}
A  packing is {\em reasonable} if,
except for
those bins that arrive when
there are no items of size $L$ remaining or at most two items of size $S$ remaining,
every bin $b$ it uses contains
\begin{itemize}
\item
one item of size $S$, if $size(b) \in [S,L-1]$,
\item
one item of size $L$, if $size(b) = L$,
\item
two items of size $S$ or one item of size $L$, if $size(b) \in [L+1,L+S-1]$,
\item
one item of size $S$ and one item of size $L$, if $size(b) = L+S$, and
\item
three items of size $S$ or one item of size $S$ and one item of size $L$,
if $size(b) \in [L+S+1,2L-1]$.
\end{itemize}
\end{definition}
Note that a 
feasible packing with no bad bins is reasonable.

\begin{corollary}
\label{reasonable}
For any set of items and any sequence of bins, there exists
an optimal packing that is reasonable.
\end{corollary}

From now on, we will restrict attention to
reasonable packings.

Given a set of items and a sequence of bins,
two reasonable
packings of these bins (where different subsets of the items
might be packed) can differ 
as to whether they
use one item of size $L$ or two items of size $S$ in certain bins.
Therefore,  the numbers of \Ss and
\Ls they do not assign may differ. However, if both
have at least one \Sitem and at least one \Litem{} available,
the set of bins they have used is the same and
there is a simple invariant relating the numbers
of available \Ss and available \Ls they have.

\begin{lemma}
Given a set of items, consider two
reasonable packings of the same sequence of bins, which might pack different
subsets of the items.
Suppose that before bin $b$, each
has at least one item of size $S$ and 
at least one item of size $L$ available.
Then immediately after bin $b$ has been packed, the sum of the number of items of size $S$
available plus twice the number of items of size $L$ available is the same for both.
\label{lemma-reasonable}
\end{lemma}

\begin{proof}
Consider any bin $b$ in the sequence $\sigma$ 
and suppose that immediately before bin $b$ is packed,
each 
packing has items of both size $S$ and $L$ available
and the sum of the number of items of size $S$
available plus twice the number of items of size $L$ available is the same for both.
Note that this is true initially, since all items are still available
for both.
Since both
are reasonable, either bin $b$ is filled the same way in both
or in one of these 
packings, bin $b$ contains two more items of size $S$ and one less item of size $L$
than the other. Thus the claim remains true immediately after bin $b$ is packed.
\end{proof}

For any sequence of bins $\sigma$ and any nonnegative integers $s$ and
$\ell$,
let $OPT(\sigma,s,\ell)$ denote the cost of an optimal packing
of $s$ items of size $S$
and $\ell$ items of size $L=2S-1$ using $\sigma$. 
This must be at least the sum of the sizes of all the items.

\begin{proposition}
For all sequences of bins $\sigma$ and all integers $s,\ell \geq 0$,
$OPT(\sigma,s,\ell) \geq sS+\ell L$.
\label{optspacelower}
\end{proposition}

Given any optimal packing for a set of items,
a packing for a subset of these items
can be obtained by removing the additional items from bins, starting
from the end.

\begin{proposition}
For all sequences of bins $\sigma$ and all integers $0 \leq s' \leq s$
and $0 \leq \ell' \leq \ell$,
$OPT(\sigma,s',\ell') \leq OPT(\sigma,s,\ell)$.
\label{optfewer}
\end{proposition}

For any sequence of bins $\sigma$
and any nonnegative integers $s$ and $\ell$,
let $R(\sigma,s,\ell)$ denote the maximum cost of any
reasonable packing
of $s$ items of size $S$
and $\ell$ items of size $L = 2S-1$
into $\sigma$.

When all items have the same size, all 
algorithms,
including \opt, behave exactly the same.

\begin{proposition}
For all sequences of bins $\sigma$ and all integers $s',\ell' \geq 0$,
$R(\sigma,s',0) = OPT(\sigma,s',0)$
and 
$R(\sigma,0,\ell')$ $= OPT(\sigma,0,\ell')$.
\label{oneitem}
\end{proposition}

Suppose that $R$ and \opt both run out of \Ls at the same time or
they both run out of \Ss at the same time.
If they have items of the other size remaining, then,
by Lemma \ref{lemma-reasonable}, they have the same number remaining.
By feasibility, they have used the same set of bins and, by Proposition \ref{oneitem},
they will use the same set of bins for the remaining items.
Thus, they have the same cost.

The following four lemmas describe the relationship between the costs
incurred by $R$ and \opt when one of them has run out of
one size of items. We begin with the case where \opt is the first to
run out of something and it runs of \Ls.

\begin{lemma}
\label{OptNoLarge}
For all sequences of bins $\sigma$ and all integers $s',\ell' \geq 0$,\\
$R(\sigma,s',\ell') \leq \OPT(\sigma,s'+2\ell',0) + \ell'(2S-3).$
\end{lemma}
\begin{proof}
by induction on $s'$ and $\ell'$.

If $\ell' =0$, then, 
by Proposition~\ref{oneitem}, $R(\sigma,s',0) = OPT(\sigma,s',0)$.

If $s'=0$, then any
packing
puts one item of size $L$ into each bin that it uses.
By Proposition \ref{optspacelower}, $OPT(\sigma,2\ell',0) \geq 2\ell' S$.
Since each bin in $\sigma$ has size at most 
$4S-3$, it follows that $R(\sigma,0,\ell') \leq \ell'(4S-3) \leq OPT(\sigma,2\ell',0) + \ell'(2S-3)$.

Let $s',\ell' \geq 1$ and suppose the claim is true for $s''$ and $\ell''$,
if $0 \leq s'' < s'$
or $0 \leq \ell'' < \ell'$.
Let $\sigma$ be any sequence of bins,  let $B \leq M = 4S-3$ be the size
of the first bin in $\sigma$, and let $\sigma'$ be obtained from $\sigma$
by removing
its first bin. Since $s',\ell' \geq 1$, it follows that $s'+2\ell' \geq 3 \geq 
\lfloor B/S \rfloor$. Note that
$OPT(\sigma,s'+2\ell',0) = B + OPT(\sigma',s'+2\ell'-\lfloor B/S \rfloor,0)$,
because \opt packs $\lfloor B/S \rfloor$ items of size $S$ in the first bin.

Consider any 
algorithm.
If it packs the first bin with only items of size $S$,
then, it packs $\lfloor B/S \rfloor$ items into that bin
and the total space it uses is at most
$B + R(\sigma',s'-\lfloor B/S \rfloor,\ell')$,
which, by the induction hypothesis, is at most
$B + OPT(\sigma',s'+2\ell'-\lfloor B/S \rfloor,0) + \ell'(2S-3) =
OPT(\sigma,s'+2\ell',0) + \ell'(2S-3)$.
So assume that, in the first bin,  the algorithm packs one item of size $L$
 plus possibly one item of size $S$.

If $B < L + S$, then the algorithm packs no items of size $S$
into the first bin and uses at most $B + R(\sigma',s',\ell'-1)$ space.
By the induction hypothesis,
$R(\sigma',s',\ell'-1) \leq OPT(\sigma',s'+2\ell'-2,0) + (\ell'-1)(2S-3)$.
Since \opt packs at most two items of size $S$ into the first bin,
$B + OPT(\sigma',s'+2\ell'-2,0) \leq OPT(\sigma,s'+2\ell',0)$. 
Hence, the space used by the algorithm is at most
$B +  R(\sigma',s',\ell'-1) \leq OPT(\sigma,s'+2\ell',0) + \ell'(2S-3)$.

Otherwise, $L+S \leq B \leq M = 2L -1$
and the algorithm also packs one item of size $S$ into the first bin.
Then the space used by the algorithm is at most
$B + R(\sigma',s'-1,\ell'-1) \leq B + OPT(\sigma',s'-1+2\ell'-2,0) + (\ell'-1)(2S-3)$,
by the induction hypothesis.
Since \opt packs at most three items of size $S$
into the first bin,
$B + OPT(\sigma',s'+2\ell'-3,0) \leq OPT(\sigma,s'+2\ell',0)$.
Hence,  the space used by the algorithm is at most
$B+ R(\sigma',s'-1,\ell'-1) \leq OPT(\sigma,s'+2\ell',0)+ \ell'(2S-3)$.

It follows that, in all cases, $R(\sigma,s',\ell') \leq OPT(\sigma,s'+2\ell',0) + \ell'(2S-3)$.
\end{proof}

Now, we consider the case where \opt is the first to
run out of something and it runs of \Ss. 
In this case, 
after this point, the worst sequence for the reasonable packing has bins of size $L$ followed by one bin of size $M$.

\begin{lemma}
\label{OptNoSmall}
For all sequences of bins $\sigma$ and all integers $s',\ell' \geq 0$,
if $2k = s' + 2\ell'$, then
$R(\sigma,s',\ell') \leq \OPT(\sigma,0,k) + (s'+\ell'-k-1)L + M.$
\end{lemma}

\begin{proof}
 If $s'=0$,
then $k = \ell'$ and
 the lemma follows from Proposition~\ref{oneitem}.
So, we assume that $s'\geq 1$.
Let $\ell'$ and $k$ be such that $2k=s'+2\ell'$. Note that $s'$ is even, so $s' \geq 2$.
Let $\sigma$ be a
sequence of bins that maximizes $R(\sigma,s',\ell') - OPT(\sigma,0,k)$.
Suppose that, among all such sequences,
$\sigma$  has the smallest capacity (i.e. the smallest sum of bin sizes).

Consider a reasonable packing with cost $R(\sigma,s',\ell')$
for
$\sigma$.
If $\sigma$ has any bins of size in $[S,L-1]$,
they are not used by \opt, which has no small items.
By the minimality of the capacity of $\sigma$,
they are used by the reasonable packing 
to pack one item of size $S$. 
All other bins have size at least $L$
and they are used by both the reasonable packing and \opt, unless 
one of them has already packed all its items.

We begin by making a number of observations about the last bin in $\sigma$,
which all follow from the definition of $\sigma$.

The last bin in $\sigma$
must be used by either \opt or the reasonable packing; otherwise,
removing it would decrease the capacity of the sequence without changing the total space
used by either packing.

If the last bin is only used by the reasonable packing, it must have size $M$;
otherwise, replacing it by a bin of size $M$ would increase the total space used by the
reasonable packing without changing the total space used by \opt.

If the last bin is only used by \opt, 
(to pack one item of size $L$),
it must have size $L$;
otherwise, replacing it by a bin of size $L$
would decrease the total space used by \opt
without changing the total space used by the reasonable packing.

If the last bin is used by \opt and the reasonable packing
to each pack one item, it must have size $L$;
otherwise, replacing it by a bin of size $L$
would decrease the capacity of the sequence
while changing the total space used by \opt and
the reasonable packing by the same amount.

Finally,
the last bin cannot be used by the reasonable packing
to pack more than one item;
otherwise, it must have size between $2S$ and $M$. If it is replaced
by a bin of size $L$ followed by a bin of size $M$,
the total space used by the reasonable packing would increase, and the
space used by \opt would not.

Now suppose that $\sigma$ has a bin, other than its last bin,
that does not have size $L$.
Let $b$ be the last such bin.
Let $s''$ and $\ell''$ be the number of items of size $S$ and $L$, respectively,
that the reasonable packing has remaining immediately after packing bin $b$.
These items are packed one per bin in each of the next $s''+\ell''$ bins.
Since the packing is reasonable, the items of size $L$ are packed before the
items of size $S$.
If $k''$ is the number of items (of size $L$) that \opt has remaining immediately
after bin $b$, they are packed one per bin in each of the next $k''$ bins.

If $size(b) \leq L-1$, then the reasonable packing packs a single item of size $S$ in this bin.
Increasing the size of $b$ to $L$ does not change which bins the reasonable packing uses:
it still packs one item in each of the $1+s''+\ell''$ bins starting with $b$, since all of them,
except possibly the last, have size $L$.
This increases the total space used by the reasonable packing.
However, the total space used by \opt remains unchanged.
This is because, if \opt has any items remaining
immediately before bin $b$ is packed, they are each packed in a bin of size $L$.
This contradicts the definition of $\sigma$.
Since $size(b) \neq L$, it follows that
$size(b) > L$.

If bin $b$ is only used by \opt (to pack one item of size $L$),
then replacing it by a bin of size $L$ would decrease the total
space used by \opt without changing the total space used by the
reasonable packing, contradicting the definition of $\sigma$.
Therefore the reasonable packing packs at least one item in bin $b$.

Suppose that \opt uses bin $b$. If the reasonable packing
packs one item in bin $b$,
then decreasing the size of bin $b$ to $L$
would decrease the capacity of the sequence
while changing the total cost incurred by \opt and
the reasonable packing by the same amount,
contradicting the definition of $\sigma$.
Hence the reasonable packing packs $h > 1$ items in bin $b$.
Decreasing the size of bin $b$ to $L$ and adding $h-1$ bins of size $L$
immediately following bin $b$
decreases the total space used by \opt by $size(b) - L > 0$,
since \opt still packs each item following bin $b$ in a bin of size $L$.
It also increases the total space used by the reasonable packing by
$hL - size(b) > 0$,
since
it packs one item in each bin,
from bin $b$ onward.
This contradicts the definition of $\sigma$.
Therefore, \opt does not use bin $b$.
Since $size(b) > L$, \opt runs out of items before bin $b$.

We now show that, at the point \opt runs out of items,
the reasonable packing only has items of size $S$.

Consider the last bin $b'$ in $\sigma$ prior to which
\opt has at least one item (of size $L$) available
and the reasonable packing has items of both size $S$ and size $L$ available.
Then, up to including when bin $b'$ is packed,
we claim that
twice the number of \Ls that \opt has available is at least
the number of \Ss plus twice the number of \Ls that
the reasonable packing has available.
In particular, it is true before the first bin of $\sigma$ is packed, since $2k = s' + 2\ell'$.
For each bin of size at least $L$, up to and including bin $b'$,
\opt packs one \Litem and the reasonable algorithm packs at least one \Litem
or at least two \Ss, so the inequality remains true.
Moreover, \opt packs no items in bins whose size is less than $L$,
so such bins do not cause the inequality to become false.

It follows that if \opt runs out of items immediately after bin $b'$,
then the reasonable packing does too.
Since \opt does not use bin $b$, but the reasonable packing does,
\opt runs out of items prior to the reasonable packing,
so this is impossible.
If the reasonable packing runs out of \Ss immediately after bin $b'$,
then \opt has at least as many \Ls left as the reasonable packing
and, hence, runs out of items at or after the reasonable packing,
so there is a contradiction again.

Therefore, immediately
after
bin $b'$ and, hence, immediately before
bin $b$, the reasonable packing only has \Ss left.
Since bin $b$ is not the last bin, the reasonable packing is feasible,
and $size(b)>L$, the reasonable packing
has at least two items in $b$.
Replace bin $b$ by a bin of size
$size(b)-S \geq S$ followed by a bin of 
size $L$. Then the reasonable packing must pack one fewer item of size $S$ in
the first of these bins and  one item of size $S$ in the second.
This adds $L-S$ to the the total space used by the reasonable packing,
but decreases
the total space used by \opt,
contradicting the definition of $\sigma$. 
Therefore, bin $b$ does not exist.

Thus, we may assume that $\sigma$ consists of $s'+\ell'-1 \geq k$ bins of 
size $L$,
followed by at most one bin of size $M$.
Then 
$OPT(\sigma,0,k) = kL$ and $R(\sigma,s',\ell')
\leq (s'+\ell'-1)L + M$,
so $R(\sigma,s',\ell') - OPT(\sigma,0,k) 
\leq (s'+\ell'-k-1)L + M$.
\end{proof}

Next, we consider the case where $R$ is the first to
run out of something and it runs of \Ls.

\begin{lemma}
\label{AlgNoLarge}
For all sequences of bins $\sigma$ and all integers $s',\ell' \geq 0$,
if $k \geq s'+2\ell'$, then
$R(\sigma,k,0) \leq \OPT(\sigma,s',\ell') + (k-s'-\ell'-1)L +M$.
\end{lemma}

\begin{proof}
by induction on $s'$ and $\ell'$.

If $\ell'=0$, then
the first $s'$ items of a reasonable packing are packed the same as \opt would pack them.
Since 
the packing is feasible,
each bin that contains one of the remaining $k-s'$ items of size $S$
has
at most $S-1$
empty space,
except for the last such bin, which has
at most $M-S$
empty space.
Thus $R(\sigma,k,0) \leq \OPT(\sigma,s',0) + (k-s')S + (k-s'-1)(S-1)+ M-S = \OPT(\sigma,s',0) + (k-s'-1)L +M$.

Next consider $s'=0$.
Any 
packing of $k$ items of size $S$ has
at most $S-1$
empty space
in each bin
that it uses, except the last, which has
at most $M-S$
empty space.
Since it uses at most $k$ bins,
its cost is at most $kS + (k-1)(S-1) + M-S
=(k-1)L +M$.
By Proposition \ref{optspacelower}, $\OPT(\sigma,0,\ell') \geq \ell' L$, so
$R(\sigma,k,0) \leq (k-1)L+ M  = \ell' L+(k-\ell'-1)L+M \leq
\OPT(\sigma,0,\ell') +(k-\ell'-1)L+ M$.

Now, let $s',\ell' > 0$ and suppose the claim is true if $s'' < s'$ or $\ell'' < \ell'$.
Let $\sigma$ be any sequence of bins, and
let $B \leq M = 4S-3$ be the size of its first bin.
Since $k \geq s'+2\ell' \geq 3$, any reasonable packing of $k$ items
of size $S$ puts $\lfloor B/S \rfloor \leq 3$
items into this bin,
so has cost at most $B + R(\sigma', k-\lfloor B/S \rfloor,0)$,
where $\sigma'$ is the sequence $\sigma$ without its first bin.

If \opt packs the first bin with only items of size $S$,
then, it packs $\lfloor B/S \rfloor \leq s'$ items into that bin
and the total space it uses is $B + \OPT(\sigma',s'-\lfloor B/S \rfloor,\ell')$.
By the induction hypothesis,
$B+R(\sigma', k-\lfloor B/S \rfloor,0) \leq  B+
\OPT(\sigma',s'-\lfloor B/S \rfloor,\ell') +  (k-\lfloor B/S\rfloor-(s'-\lfloor B/S \rfloor)-\ell'-1)L +M
= \OPT(\sigma,s',\ell') + (k-s'-\ell'-1)L+M$.

Otherwise \opt packs the first bin with one item of size $L$ and
$\lfloor (B-L)/S \rfloor \geq \lfloor B/S \rfloor -2$
items of size $S$.
Since $k - \lfloor B/S \rfloor \geq s'+2\ell' - \lfloor (B-L)/S \rfloor -2 = (s'-\lfloor (B-L)/S \rfloor) +2(\ell'-1)$,
it follows 
by the induction hypothesis that $B+R(\sigma', k-\lfloor B/S \rfloor,0) \leq 
B+\OPT(\sigma', s'-\lfloor (B-L)/S \rfloor,\ell'-1) +
(k-\lfloor B/S \rfloor-(s'-\lfloor B/S-2 \rfloor)-(\ell'-1)-1)L + M
\leq \OPT(\sigma,s',\ell') +(k-s'-\ell'-1)L +M$. Note that this last case
is the only place that the possibility of $k$ being greater than
$s'+2\ell'$ is used.

Hence, $R(\sigma,k,0) \leq B + R(\sigma', k-\lfloor B/S \rfloor,0) \leq \OPT(\sigma,s',\ell')+(k-s'-\ell'-1)L +M$.
\end{proof}

Finally, we consider the case where $R$ is the first to
run out of something and it runs of \Ss.

\begin{lemma} 
\label{AlgNoSmall}
For all sequences of bins $\sigma$ and all integers $s',\ell' \geq 0$,
if $2k \geq s'+2\ell'$,
then
$R(\sigma,0,k) \leq \OPT(\sigma,s',\ell')+\min\{ 0,\ell'-s'\} S +(k-\ell')M$.
\end{lemma}

\begin{proof}
by induction on $s'$ and $\ell'$.

If $s'=0$, then
the first $\ell'$ items are packed the same way in a reasonable packing as \opt would pack them.
Each of the remaining $k-\ell'$ items of size $L$ are packed one per bin using bins of size at most $M$,
for a total cost at most $\OPT(\sigma,0,\ell') +(k-\ell')M
= \OPT(\sigma,0,\ell')+\min\{ 0,\ell'-s'\}S +(k-\ell')M$, since $\ell' - s' \geq 0$.

If $\ell'=0$, then, by Proposition \ref{optspacelower},
$\OPT(\sigma,s',\ell') \geq s'S+ \ell'L = \max\{ 0,s'-\ell'\} S$,
so $\OPT(\sigma,s',\ell') + \min\{ 0,\ell'-s'\}S \geq 0$.
In a reasonable packing, each of the $k$ items of size $L$
are packed one per bin using bins of size at most $M$
for a total cost at most $kM \leq \OPT(\sigma,s',\ell')+\min\{ 0,\ell'-s'\}S 
 +(k-\ell')M$.

Now, let $s',\ell' \geq 1$ and suppose the claim is true if $s'' < s'$ or $\ell'' < \ell'$.
Let $\sigma$ be any sequence of bins,
let $B \leq M = 4S-3$ be the size of its first bin, and let $\sigma'$
denote the same sequence as $\sigma$, but with the first bin removed.

If $B < L$, then a reasonable packing leaves the first bin of $\sigma$ empty,
since it has no items of size $S$. Hence, it uses space at most $R(\sigma',0,k)$.
Since $2k \geq s'+2\ell' \geq (s'-1)+2\ell'$, the induction hypothesis
implies that $R(\sigma',0,k) \leq \OPT(\sigma',s'-1,\ell')+\min\{ 0,\ell'-(s'-1)\}S +(k-\ell')M$.
\opt packs one item of size $S$ into the first bin of $\sigma$,
so $\OPT(\sigma,s',\ell') = B + \OPT(\sigma',s'-1,\ell')$.
Thus, since $S \leq B$, it follows that $R(\sigma',0,k) \leq
\OPT(\sigma',s'-1,\ell') + B- S+\min\{ 0,\ell'-(s'-1)\}S +(k-\ell')M
\leq \OPT(\sigma,s',\ell')+\min\{ 0,\ell'-s'\}S  +(k-\ell')M$.

Otherwise, $B\geq L$. Then, after packing the first bin of $\sigma$,
a reasonable packing has $k-1$ items (of size $L$) remaining.
Suppose \opt has $s''$ items of size $S$ and $\ell''$ items of size $L$ remaining
after packing this bin.
Hence, by the induction hypothesis, to pack its $k-1$ \Ls, any reasonable packing uses space at most $R(\sigma', 0,k-1) \leq \OPT(\sigma',s'',\ell'')+\min\{ 0,\ell''-s''\}S +(k-1-\ell'')M$,
since $2(k-1) \geq s'+2(\ell'-1) = (s'-2)+2\ell'  \geq s''+2\ell''$.
Therefore, $R(\sigma,0,k) = B + R(\sigma', 0,k-1)
\leq B + \OPT(\sigma',s'',\ell'')+\min\{ 0,\ell''-s''\}S +(k-1-\ell'')M
= \OPT(\sigma,s',\ell')+\min\{ 0,\ell''-s''\}S +(k-1-\ell'')M$.

It remains to be shown that
\begin{equation}
\label{eqn1}
\min\{ 0,\ell''-s''\}S +(k-1-\ell'')M \leq \min\{ 0,\ell'-s'\}S +(k-\ell')M.
\end{equation}
Since \opt is reasonable and $B \leq M = 2L - 1 = 4S-3$, either $\ell'' = \ell'-1$ and
$s'' \geq s' - 1$ or $\ell'' = \ell'$ and  $s'' \geq s'-3$.
In the first case, \ref{eqn1} follows since  $\ell''-s'' \leq (\ell' - 1) - (s' - 1) = \ell' - s'$ and $k -1 - \ell'' = k - \ell'$,
so $\min\{ 0,\ell''-s''\}S +(k-1-\ell'')M \leq \min\{ 0,\ell'-s'\}S  +(k-\ell')M$.
In the second case, \ref{eqn1} follows since $\ell''-s'' \leq \ell' -s' +3$ and $k-1 - \ell'' = k - \ell' - 1$,
so $\min\{ 0,\ell''-s''\}S +(k-1-\ell'')M \leq \min\{ 0,\ell'-s'\}S +3S +(k-\ell')M - M
\leq \min\{ 0,\ell'-s'\}S +(k-\ell')M$, since $3S \leq 4S-3 = M$, provided $S \geq 3$.
Thus, asymptotically (for $S\geq 3$), $R(\sigma,0,k) \leq \OPT(\sigma,s',\ell')+\min\{ 0,\ell''-s''\}S +(k-1-\ell'')M
\leq \OPT(\sigma,s',\ell')+\min\{ 0,\ell'-s'\}S +(k-\ell')M$.
\end{proof}

\section{Simple Non-Optimal Algorithms}
\label{sec:simple}
It is helpful to understand why simple reasonable algorithms
are not optimal for the Restricted Grid Scheduling problem.
The following examples show why a number of natural
candidates do not work well enough.

\paragraph{Example 1} 
Consider the reasonable algorithm that
always uses \Ss, when there  is a choice. 
Let $s = 2\ell$ and let $\sigma$ consist of
$\ell$ bins of size $2S$, followed by $s$ bins of size $S$,
and then $\ell$ bins of size $\maxbinsize$.
For this instance, the algorithm has
a \perfratio{} of
$\frac{\ell\cdot 2S+\ell\cdot \maxbinsize}
{\ell \cdot 2S + s \cdot S}=\frac{3}{2} - \frac{3}{4S}$, which is greater
than $5/4$ for large $S$.

\vspace*{-1ex}
\paragraph{Example 2} 
Consider the reasonable algorithm that
always uses \Ls, when there  is a choice. 
Let $s = 2\ell$ and let $\sigma$ consist of
$\ell$ bins of size $2S$, followed by $s-1$ bins of size $L$
and then one bin of size $\maxbinsize$.
For this instance, the algorithm has a \perfratio{} of
$\frac{\ell\cdot 2S+(s-1)\cdot L+ \maxbinsize}
{\ell \cdot 2S+\ell \cdot L}
=\frac{3}{2}  + \frac{1}{2\ell} - \frac{3/\ell+1}{2(4S -1)}$,
which is also greater
than $5/4$ for large $S$.

\medskip

For the two instances considered above, the reasonable algorithm
which alternates between
using two \Ss and one \Litem, when it has a choice, would do well, achieving a 
performance ratio of $5/4$.
However, it does not do as well on other instances.

\vspace*{-1ex}
\paragraph{Example 3} 
Let $2s=3\ell$ and let $\sigma$ consist of
$\ell$ bins of size $2S$, followed by $s$ bins
of size $S$ and $\ell/2$ bins of size $\maxbinsize$.
For this instance, the algorithm which alternates between
using two \Ss and one \Litem, when it has a choice, has a
\perfratio{} of $\frac{\ell \cdot 2S + (s- \ell) \cdot S+
\ell/2 \cdot \maxbinsize}
{\ell \cdot 2S+s\cdot S}
=1+\frac{\frac{\ell}{2}\cdot(\maxbinsize-2S)}{(2\ell+s)\cdot S}
=1+\frac{2}{7}-\frac{3}{7S}$, which is larger than $5/4$ for
$S$ sufficiently large.

\medskip

As the above examples partially illustrate, once either the online
algorithm or \opt has run out of one type of item (\Ss or \Ls), the
adversary can give bins which make the online algorithm waste
a lot of space. The algorithm we present in the next section
aims to postpone this
situation long enough to get a good ratio.

When $3s = 2 \ell$, the lower bound is in Case II.
It starts by giving the algorithm
$\lfloor s/2 \rfloor$ 
bins of size $2S$
and it considers the number $k$ of these bins in which
the algorithm packs two \Ss. 
The best the algorithm can do is to have
$k \approx \max\{ \lfloor s/2 \rfloor - s/8 - \ell/4 +1, \lfloor s/2 \rfloor - s/3 +1\}
=\lfloor s/2 \rfloor - s/3 +1 \approx s/6$, so that the same ratio
is obtained for both subcases.
This is the same as using one \Litem 
twice as often as two \Ss, when it has a choice.

\paragraph{Example 4} 
Let $3s = 2\ell$.
Consider the reasonable algorithm which uses one \Litem 
twice as often as two \Ss, when it has a choice. 
Let $\sigma$ consist of
$\ell$ bins of size $2S$, followed by $s$ bins
of size $S$ and $\ell/3$ bins of size $\maxbinsize$.
For this instance, the algorithm
packs $\ell/3 = s/2$ bins of size $2S$ with two \Ss,
so it has a
\perfratio{} of $\frac{\ell \cdot 2S + (\ell/3) \cdot M}
{\ell \cdot 2S+s\cdot S}
= 1+\frac{\frac{s}{2}\cdot\maxbinsize-sS}{(2\ell+s)\cdot S}
=1+\frac{S -3/2}{4 \cdot S}
=\frac{5}{4} -\frac{3}{8S}$, which
exceeds the lower bound of $1 + \frac{2}{3 + 6\ell/s}  = \frac{7}{6}$
for $S$ sufficiently large.

This example shows that it is not always possible to achieve the lower bound
by choosing to use two \Ss instead of one \Litem a fixed proportion of the time.
After $s/6$ bins have been packed with two \Ss,
it turns out to be better to always use an \Litem, when there is a choice.
This motivates our use of a second phase to obtain the best competitive ratio.
The number of \Ss that the algorithm has packed, in this case $s/6 \times 2 = s/3$,
indicates when the first phase is done.

\section{\newalg}
\label{sec:upper}
In Figure \ref{fig:newalg},
we present a reasonable algorithm, \newalg,
for the Restricted Grid Scheduling problem.
It is asymptotically optimal:
the competitive ratio matches the lower bound in Section~\ref{sec:lower}
for all three
ranges of the ratio $s/\ell$
of the initial numbers of \Ss and \Ls.
Not surprisingly,
\newalg has two phases.

In the first phase, it attempts to balance
the number of \Ss and the number of \Ls it uses, aiming for the ratio indicated
by the lower bound.
When it receives bins where it has a choice of using
one \Litem or two \Ss in part or all of that bin (i.e., bins with sizes
in the ranges $[2S,3S-2]$ and $[3S,4S-3]$), it
uses one \Litem in a certain fraction of them
(and increments the variable \LB)
and uses two \Ss in
the remaining
fraction (and increments \SB).
The fraction
varies 
depending on the original
ratio $s/\ell$:
at least 2, between 2 and 6/5, and less than 6/5.
It
is enforced by a macro called
UseLs, which indicates that an \Litem should be packed if and
only if $\LB \leq r\SB$, where $r$ is the target
ratio of \LB to \SB.
For example, when
the number of bins containing one \Litem should be within $1$ of
the number of bins containing two \Ss,
we have  UseLs = (\LB $\leq$ \SB). Both \LB and \SB start
at zero, so, in this case, \newalg starts by choosing
to use one \Litem 
and then alternates.
Note that \SB and \LB do not change after Phase 1 is completed.

In the middle range for the ratio $s/\ell$, there 
are two different fractions used.
The first fraction is used until
\SB reaches
a specific value. Afterwards,
its choices alternate.
To do so, for the rest of Phase 1, it 
records the number of times it chooses to pack two \Ss
in a bin and the number of times it chooses to pack one \Litem
in \SBt and \LBt,
respectively. 
The variables \SBt and \LBt are also zero initially.

\newalg uses \SC and \LC throughout the algorithm
to keep track of the total
numbers of \Ss and \Ls it has packed, whether or not it had a choice.
(Specifically, \SC is incremented every time an item of size $S$ is packed
and \LC is incremented every time an \Litem is packed.)
It continues with Phase 1 until 
it has packed a certain number of \Ss or \Ls (depending on 
the relationship between $s$ and $\ell$).
For each of the three ranges of $s/\ell$,
we define a different
condition for ending Phase 1. In Phase 2, only \Ss or
only \Ls are packed
where a reasonable algorithm has a choice,
depending on whether one would expect an excess of \Ss
or \Ls, given the ratio
$s/\ell$.

\paragraph{An example of an execution of \newalg} 
Consider running \newalg
when $\ell = s/2$ is even and the sequence of bins consists of
$3\ell/2$ bins of size $2S$, followed by $\ell$ bins of size $S$, and then $\ell/2$ bins of size $M$.
During Phase 1, \newalg alternates between packing an item
of size $L$ and  two \Ss, starting with an item of size $L$, because
the condition ($\LB \leq \SB$) will alternate between being true and false.
Phase 1 ends when $ \ell/2$ \Ls have been packed.
In Phase 2, the remaining $\ell/2 $ \Ls and $s/2+2$ \Ss are packed.
These \Ss are packed in the last $\ell/2 +1$ bins of size $2S$.
\newalg leaves the bins of size $S$ empty, because no \Ss remain.
It puts one item of size $L$ in 
each bin of size $M$. Note that \newalg has cost $(3\ell/2) \cdot 2S + (\ell/2) \cdot M = 5\ell S -3\ell/2$.
In contrast, \OPT packs all the \Ls in  bins of size $2S$, packs half the \Ss in bins of size $2S$, and packs
the other half in the bins of size $S$. Its cost is
$3\ell/2 \cdot 2S +\ell S = 4\ell S$.  The ratio of their costs is
$\frac{5\ell S-\frac{3\ell}{2}}{4\ell S}$.

The definition of the algorithm implies the following relationships.

\begin{observation}
\label{binsvscounts}
$\LB + \LBt \leq \LC$ and $2(\SB + \SBt) \leq \SC$.\\
If $\ell \leq s/2$ or $5s/6 < \ell$, then $\SBt = \LBt = 0$.
\end{observation}

The definitions of the end of Phase 1 for the various ranges of $s/\ell$
imply these inequalities.

\begin{lemma}
\label{Sbininv2}
If $s/2 < \ell \leq 5s/6$, then $\SB
+ \SBt
\leq 3s/8 - \ell/4 + 1$.\\
If $5s/6 < \ell $, then $\SB \leq s/6 + 1$.
\label{Sbininv3}
\end{lemma}

The following simple invariants can be proved inductively.

\begin{lemma}
\label{loopinv}
If $\ell \leq s/2$, then $\SB  \leq \LB \leq \SB +1$.\\
If $s/2 < \ell \leq 5s/6$, then
$\SBt  \leq \LBt \leq \SBt +1$ and\\
$\lfloor c (\SB - 1) \rfloor +1 \leq \LB \leq
\lfloor c \ \SB \rfloor +1$,
where
$1 < c = \frac{10\ell-3s}{s+ 2\ell} \leq 2$.\\
If $5s/6 <  \ell$, then
$2 \SB  \leq \LB \leq 2 \SB +1$.
\end{lemma}

We are now ready to prove the following:
\begin{theorem}
The Restricted Grid Scheduling problem
has competitive ratio at most
$$\left \{ \begin{array}{ll}
1 + \frac{1}{2 + s/\ell}
& \mbox{ if }  \ell \leq s/2,\\

1 + \frac{1}{4}
& \mbox{ if }  s/2 < \ell \leq 5s/6, \mbox{ and}\\

1 + \frac{2}{3 + 6\ell/s}
& \mbox{ if }  5s/6< \ell.
\end{array} \right .$$
\end{theorem}

\begin{proof}
We analyse the performance of \newalg as compared to the cost of an optimal reasonable packing on
an arbitrary sequence, $\sigma$, using the three different ranges
for the relationship between $s$ and $\ell$ in the lower bound. 

Both \newalg and \opt use exactly the same bins until one of
them runs out of either \Ls or \Ss. Thus,
there are four
cases to consider.

\begin{figure}[H]
\begin{center}
\begin{minipage}{15cm}
{\bf macro} Phase1done = $\left\{ \begin{array}{l@{\;\;}l}
\LC \geq \lfloor\ell/2\rfloor & \mbox{if } \ell \leq s/2\\
\SC \geq \lfloor 3s/4-\ell/2\rfloor & \mbox{if } s/2 < \ell \leq 5s/6\\
\SC \geq \lfloor s/3\rfloor & \mbox{if } 5s/6 < \ell
\end{array}\right .$ \\
 \\
 \\
{\bf macro} UseLs = $\left\{ \begin{array}{l@{\;\;}l}
\LB \leq \SB & \mbox{if } \ell \leq s/2\\
\!\!\!\begin{array}[t]{l}
{\bf if}\  (\SB < \lfloor (s+2\ell)/16\rfloor)\\
{\bf then}\ \LB \leq (10\ell-3s)\SB/(s+2\ell)\\
{\bf else}\ \LBt \leq \SBt
\end{array}
& \mbox{if } s/2 < \ell \leq 5s/6\\
\LB \leq 2\SB & \mbox{if } 5s/6 < \ell
\end{array}\right .$ \\
 \\
$\SC$ \% counts the number of \Ss packed; initially 0\\
$\LC$ \% counts the number of \Ls packed; initially 0\\
$\SB \leftarrow \LB \leftarrow \SBt \leftarrow \LBt \leftarrow 0$\\
 \\
{\bf for} each arriving bin $b$\\
\tab {\bf if} only one type of item still remains {\bf then}
use as many items as fit in $b$\\
\tab {\bf else if} $size(b)\in [S,2S-2]$ {\bf then} use 1 \Sitem\\
\tab {\bf else if} $size(b)=2S-1$ {\bf then} use 1 \Litem\\
\tab {\bf else if} $size(b)=3S-1$ {\bf then} use 1 \Litem and 1 \Sitem\\
\tab {\bf else if} only one \Sitem still remains\\
\tab\tab\;\; {\bf then} use 1 \Litem\\
\tab\tab\tab\;\;\;\;\;\; {\bf if} remaining space in $b$ is at least $S$
{\bf then} use 1 \Sitem\\
\tab {\bf else if} only two \Ss still remain and
$size(b) \in  [3S,4S-3]$\\
\tab\tab\;  {\bf then} use 1 \Litem and 1 \Sitem\\
\tab {\bf else if} ($\mbox{not Phase1done}$)\\
\tab\tab\; {\bf then} \:\% Use range determined ratio\\
\tab\tab\tab\tab\;  {\bf if} $size(b)\in [2S,3S-2]$ {\bf then}\\
\tab\tab\tab\tab\tab\tab  {\bf if} $\mbox{UseLs}$ {\bf then} use 1 \Litem\\
\tab\tab\tab\tab\tab\tab  {\bf else} use 2 \Ss\\
\tab\tab\tab\tab\;{\bf else} \:\% $size(b)\in [3S,4S-3]$\\
\tab\tab\tab\tab\tab\tab  {\bf if} $\mbox{UseLs}$ {\bf then}
 use 1 \Litem and 1 \Sitem\\
\tab\tab\tab\tab\tab\tab {\bf else} use 3 \Ss\\
\tab\tab\tab\tab\; {\bf if} $s/2 < \ell \leq5s/6$
and $\SB \geq \lfloor (s+2\ell)/16\rfloor$ {\bf then} \\
\tab\tab\tab\tab\tab\tab {\bf if} UseLs {\bf then} $\LBt+\!\!+$\\
\tab\tab\tab\tab\tab\tab {\bf else} $\SBt+\!\!+$\\
\tab\tab\tab\tab\; {\bf else}  \:\% $\ell\leq s/2$ or $\ell > 5s/6$ or $\SB < \lfloor (s+2\ell)/16\rfloor$\\
\tab\tab\tab\tab\tab\tab {\bf if} UseLs {\bf then} $\LB+\!\!+$\\
\tab\tab\tab\tab\tab\tab {\bf else} $\SB+\!\!+$\\
\tab {\bf else} \% In Phase 2\\
\tab\tab\;\;  {\bf if} $\ell \leq s/2$ {\bf then}
use as many \Ss as fit in bin $b$\\
\tab\tab\;\;  {\bf else}
use 1 \Litem\\
\tab\tab\tab\;\;\;\; {\bf if} remaining space in $b$ is at least $S$ {\bf then}
use 1 \Sitem\\
{\bf end for}
\end{minipage}
\end{center}
\caption{The algorithm \newalg}
\label{fig:newalg}
\end{figure}

\begin{enumerate}
\item \OPT runs out of \Ls at or before the point where 
\newalg runs out of anything.
\item \OPT runs out of \Ss at or before the point where 
\newalg runs out of anything.
\item \newalg runs out of \Ls before \OPT runs out of anything.
\item \newalg runs out of \Ss before \OPT runs out of anything.
\end{enumerate}

We examine each of these four cases and, within each case, we
consider three different ranges for the relationship between
$s$ and $\ell$.

Consider an  arbitrary sequence of bins $\sigma$ in which
\newalg{}  packs $s$ items of size $S$ and
$\ell$ items of size $L$.
Consider the first bin $b'$ after which
either \OPT{} or \newalg{} had only one type of item remaining.
Up to and including bin $b'$, both \OPT and \newalg have
unpacked items of both sizes.
Let $\CL$ and \CS  denote the number of \Ls and \Ss,
respectively, that \newalg has packed up to and including bin $b'$.
Let $\Spr$ and $\SBtpr$ be the values of $\SB$ and $\SBt$ immediately after bin $b'$,
so $\Spr +\SBtpr$ denotes the number of bins at or before $b'$
where \newalg had a choice and packed two \Ss instead of one of size $L$.
Similarly, 
$\Lpr +\LBtpr$ denotes the number in which \newalg
packed one \Litem instead of two items of size $S$, where
$\Lpr$ and $\LBtpr$ are the values of $\LB$ and $\LBt$ immediately after bin $b'$.
Let $\sigma'$ denote the portion
of $\sigma$ after bin $b'$.

Both algorithms use all bins up to and including bin $b'$,
since all bins have size at least $S$.
Let $X$ be the total cost of these bins.
Then $X \geq (\CS) S + (\CL) L$ and
\newalg($\sigma,s,\ell) = X +$ \newalg($\sigma', s-\CS, \ell-\CL$).

\paragraph{Case 1:
\OPT has no
\Ls after $b'$}
In this case,
Lemma \ref{lemma-reasonable} implies that,
after bin $b'$,
OPT has $(s-\CS)+2(\ell-\CL)$
items of size $S$ remaining,
so $OPT(\sigma,s,\ell) = X + OPT(\sigma',(s-\CS)+2(\ell-\CL),0)$.

Since \newalg is reasonable, Lemma~\ref{OptNoLarge}
implies that 
$\mbox{\newalg}(\sigma', s-\CS, \ell-\CL)  \leq 
\OPT(\sigma',(s-\CS)+2(\ell-\CL),0)+(\ell-\CL)(2S-3)$,
so 
$$\mbox{\newalg}(\sigma,s,\ell)  \leq \OPT(\sigma,s,\ell) +(\ell-\CL)(2S-3).$$
By Proposition~\ref{optspacelower},
$\OPT(\sigma,s,\ell) \geq sS+\ell L
\geq sL/2+\ell L = (s+2\ell)L/2$.
When computing the competitive ratio, we will choose the additive constant
to be at least $L$, so we subtract $L$ from \newalg's cost in the ratio.
Thus,
\begin{eqnarray*}
\frac{\newalg(\sigma,s,\ell)-L}{\OPT(\sigma,s,\ell)} & \leq &
\frac{\OPT(\sigma,s,\ell)+(\ell-\CL)(2S-3)-L}{\OPT(\sigma,s,\ell)}\\
& \leq & 1 + \frac{(\ell-\CL-1)(L-2)}{(s+2\ell)L/2}\\
& \leq &
1+\frac{2(\ell-\CL-1)}{s+2\ell}.
\end{eqnarray*}
It remains to bound $\ell-\CL- 1$ in each of the three ranges for the ratio
$s/\ell$.
We use the fact that, immediately after bin $b'$,
OPT has run out of \Ls, so
at least $\ell$ bins of size at least $L$ are in $\sigma$
up to and including bin $b'$.

First, consider the case when $\ell \leq s/2$.
If $\CL < \lfloor \ell/2 \rfloor$, then Phase 1 was not completed
when bin $b'$ arrived. 
Therefore,
each time a bin of size
at least $L$ arrives up to and including bin $b'$,
\newalg either packs an \Litem in it and, hence,
increments \LC, or it increments \SB.
Hence, $\CL + \Spr \geq \ell$.
By Lemma \ref{loopinv}, $\Lpr \ge \Spr$.
Since $\CL \geq \Lpr$, by Observation~\ref{binsvscounts},
it follows that $\CL \geq \ell/2 \geq \lfloor \ell/2 \rfloor$.
This is a contradiction.
Thus,  $\CL \geq \lfloor \ell/2 \rfloor$ and 
\begin{eqnarray*}
\frac{\newalg(\sigma,s,\ell)-L}{\OPT(\sigma,s,\ell)} 
& \leq & 1+\frac{2(\ell-\CL-1) }{s+2\ell} \\
& \leq & 1+\frac{\ell}{s+2\ell} =
1+\frac{1}{2+s/\ell}.
\end{eqnarray*}

So, suppose that $s/2 < \ell$. 
During Phase 1, each time a bin of size
at least $L$ arrives,  \newalg either packs an \Litem in it and, hence,
increments \LC, or it increments one of \SB or \SBt.
During Phase 2, \newalg  packs an \Litem in each such bin until
it runs out of \Ls.
Therefore,  $\CL + \Spr +\SBtpr \geq \ell$.

If $\ell \leq 5s/6$, then Lemma \ref{Sbininv2} implies that
$\ell-\CL\leq \Spr +\SBtpr \leq 3s/8-\ell/4+1$ and
\begin{eqnarray*}
\frac{\newalg(\sigma,s,\ell)-L}{\OPT(\sigma,s,\ell)} 
& \leq & 1+\frac{2(\ell-\CL-1)}{s+2\ell} 
~ \leq ~ 1+\frac{2(3s/8-\ell/4)}{s+2\ell} \\
& = & 1+\frac{s/4+\ell/2 + s/2-\ell}{s+2\ell} 
~ < ~ 1 + \frac{s/4 + \ell/2}{s+2\ell}
=  \frac{5}{4}.
\end{eqnarray*}

Similarly,  if $s< 6\ell/5$, then 
$\SBtpr = 0$, by Observation~\ref{binsvscounts}, so
$\ell-\CL\leq \Spr \leq s/6+1$,
by Lemma \ref{Sbininv3}, and
\begin{eqnarray*}
\frac{\newalg(\sigma,s,\ell)-L}{\OPT(\sigma,s,\ell)} 
& \leq & 1+\frac{2(\ell-\CL-1) }{s+2\ell} \\
& \leq & 1+ \frac{s/3}{s+2\ell} 
~ < ~ 1+\frac{2}{3+6\ell/s}. 
\end{eqnarray*}

\paragraph{Case 2:
\OPT has  no \Ss after $b'$}
In this case, Lemma \ref{lemma-reasonable} implies that,
after bin $b'$,
OPT has $k= (s-\CS)/2+(\ell-\CL)$
items of size $L$ remaining,
so $OPT(\sigma,s,\ell) = X + OPT(\sigma',0,k)$.

Since  \newalg\ is reasonable, Lemma~\ref{OptNoSmall} implies that
$$\begin{array}{l}
\newalg(\sigma,s,\ell) = X+\newalg(\sigma',s-\CS,\ell-\CL)\\
\leq X+OPT(\sigma',0,k)+ ((s-\CS)+(\ell-\CL)-k-1)L + M\\
= OPT(\sigma,s,\ell) +  (s-\CS)L/2 + M-L.
\end{array}$$
By Proposition~\ref{optspacelower},
$OPT(\sigma,s,\ell) \geq sS+\ell L\geq sL/2+\ell L$.

When computing the competitive ratio, we will choose the additive constant
to be at least $M$, so we subtract $M$ from \newalg's cost in the ratio.
Thus,
\begin{eqnarray*}
\frac{\newalg(\sigma,s,\ell)-M}{\OPT(\sigma,s,\ell)} & \leq & 
\frac{\OPT(\sigma,s,\ell) + (s-\CS)L/2 -L}
{\OPT(\sigma,s,\ell)} \\
& \leq & 1 + \frac{((s-\CS)/2-1)L}{sL/2+\ell L} \\
 & = & 1 + \frac{s-\CS-2}{s +2\ell }.
\end{eqnarray*}
It remains to bound $s-\CS-2$ in each of the three ranges for the ratio $s/\ell$.

First, suppose that $\ell\leq s/2$.
Immediately after bin $b'$,
OPT has packed at most $2\Lpr$ additional \Ss than
\newalg, so
$s \leq \CS +2\Lpr$.
When Phase 1 is finished,
$\LC = \lfloor \ell/2\rfloor$, since $\LC$ either changes by 0 or 1 each 
iteration.
After Phase 1, $\LB$ does not change.
By Observation~\ref{binsvscounts}, $\LB \leq \LC$.
Hence $\Lpr \leq \lfloor \ell/2\rfloor \leq \ell/2$.
Thus,
\begin{eqnarray*}
\frac{\newalg(\sigma,s,\ell)-M}{\OPT(\sigma,s,\ell)}
& \leq & 1 + \frac{s-\CS-2}{s +2\ell }
< 1 + \frac{2\Lpr}{s+2\ell} \\
& \leq & 1 + \frac{\ell}{s+2\ell} = 
1+\frac{1}{2+s/\ell} .
\end{eqnarray*}

Next, consider the case where $s/2 < \ell \leq 5s/6$.
If $\CS \geq 3s/4 - \ell/2 - 2$, then
$s-\CS -2 \leq 
s -3s/4 + \ell/2 =
(s+2\ell)/4$.
Otherwise, $\CS \leq 3s/4-\ell/2 -2 < \lfloor 3s/4-\ell/2  \rfloor$, so
Phase 1 is not completed 
when bin $b'$ is packed.
Therefore,
the number of bins in which \newalg\ packed one \Litem,
but OPT  packed two \Ss is at most  $\Lpr+\LBtpr$.
Since  \OPT has packed all of the \Ss,
$s \leq \CS +2\Lpr+ 2\LBtpr$.
By Observation~\ref{binsvscounts} and Lemma~\ref{loopinv},
$2\Spr + 2\SBtpr \leq \CS$,
$\Lpr \leq \lfloor (10 \ell -3s)\Spr/(s+2\ell) \rfloor +1$, and
$\LBtpr \leq \SBtpr +1$.
From the algorithm, $\Spr \leq \lfloor (s+2\ell)/16 \rfloor$,
since $\SB$ increases by at most one each iteration until it reaches
$\lfloor (s+2\ell)/16 \rfloor$ and does not increase thereafter.
Hence,
$$\begin{array}{l}
s -\CS -2 \leq 2(\Lpr + \LBtpr -1) \\ \leq  2(1+ (10\ell -3s)\Spr/(s+2\ell) + 1+ \SBtpr -1) \\
 =   2(\Spr + \SBtpr + (8\ell - 4s)\Spr/(s+2\ell)+1) \\
 \leq  \CS + 2(8\ell - 4s)/16 + 2 \\
 \leq  3s/4 - \ell/2 + \ell - s/2 = s/4 + \ell/2 = (s+2\ell)/4.
\end{array}$$
Therefore,
\begin{eqnarray*}
\frac{\newalg(\sigma,s,\ell)-M}{\OPT(\sigma,s,\ell)} \leq  1 + \frac{s-\CS-2}{s+2\ell} 
\leq 1+\frac{(s+2\ell)/4}{s+2\ell} = \frac{5}{4}.
\end{eqnarray*}

When $5s/6 < \ell$, we show that $\CS \geq \lfloor s/3 \rfloor$.
Suppose not. Then Phase 1 is not completed 
when bin $b'$ is packed, so
the number of bins in which \newalg\ packed one \Litem,
but OPT  packed two \Ss is at most  \Lpr.
Since  \OPT has packed all of the \Ss,
$s \leq \CS +2\Lpr$.
By Lemma~\ref{loopinv} and Observation~\ref{binsvscounts}, 
$\Lpr \leq 2\Spr +1 \leq  \CS +1$.
Hence $s \leq 3 \CS + 2$ and $\CS \geq (s-2)/3$.
Since \CS is an integer and $(s-2)/3 > (s-3)/3 \geq \lfloor s/3 \rfloor-1$
it follows that
$\CS \geq \lfloor s/3 \rfloor$, which is a contradiction.

Thus, $\CS \geq \lfloor s/3 \rfloor$, 
so $s-\CS -2 \leq \lceil 2s/3\rceil-2 < 2s/3$ and
\begin{eqnarray*}
\frac{\newalg(\sigma,s,\ell)-M}{\OPT(\sigma,s,\ell)}
 & \leq & 1 + \frac{s-\CS-2}{s+2\ell} \\
& < & 1+\frac{2s/3}{s+2\ell}\\
& = & 1+\frac{2}{3 +6\ell/s}.
\end{eqnarray*}

\paragraph{Case 3:
 \newalg runs out of \Ls before \OPT runs out of anything}
Suppose that, in this instance,
\newalg packs its last item of size $L$ in bin $b'$,
but, after  bin $b'$,  \OPT has $i>0$ items of size $S$
and $j>0$ items of size $L$ that are unpacked.
Then, by Lemma \ref{lemma-reasonable},
\newalg has $s -\CS = i + 2j$ unpacked items of size $S$.

Note that $j\leq \Lpr+\LBtpr$, since the only bins where \newalg had a 
choice and packed one \Litem where \OPT possibly did not are the bins
contributing to the values of $\LB$ and $\LBt$. 
By Lemma~\ref{AlgNoLarge},
$\newalg(\sigma,s,\ell) 
= X + \newalg(\sigma',i+2j,0) \leq X+ \OPT(\sigma',i,j) + (j-1)L +M$
= $\OPT(\sigma,s,\ell) + (j-1)L +M$.
By Proposition~\ref{optspacelower},
$\OPT(\sigma,s,\ell)
\geq sS+\ell L\geq sL/2+\ell L$.
In computing the competitive ratio, we will let the additive constant
be $M$, so we will subtract this value from \newalg's
cost in the ratio.

Thus,
\begin{eqnarray*}
\frac{\newalg(\sigma,s,\ell)-M}{\OPT(\sigma,s,\ell)} & \leq & 
\frac{\OPT(\sigma,s,\ell) + (j-1)L}{\OPT(\sigma,s,\ell) } \\
&= & 1+ \frac{(j-1)L}{\OPT(\sigma,s,\ell) } \\
 & \leq & 1 + \frac{(j-1)L}{sL/2+\ell L } \\
 & = & 1+\frac{2j-2}{s+2\ell} 
\end{eqnarray*}

We now bound $2j$ in each of the three cases of the algorithm.

Suppose $\ell\leq s/2$. 
Phase 1 has completed, since all \Ls have been packed by \newalg.
\Lpr is at most the value of \LC when Phase 1 completed,
which is $\lfloor\ell/2\rfloor$.
Thus, $2j\leq 2\Lpr\leq
2\lfloor\ell/2\rfloor \leq \ell$.
Hence, 
\begin{eqnarray*}
\frac{\newalg(\sigma,s,\ell)-M}{\OPT(\sigma,s,\ell)} & \leq &
1 + \frac{2j}{s+2\ell}\\
& \leq & 1 + \frac{\ell}{s+2\ell}\\
& \leq & 1 + \frac{1}{2+s/\ell}.
\end{eqnarray*}

Next, suppose that $6\ell/5 \leq s<2\ell$.
If $\CS \geq 3s/4 - \ell/2 -2$, then
$2j-2 = s- \CS -i -2 < s/4 + \ell/2= (s + 2\ell)/4$.
Otherwise, $\CS < 3s/4 - \ell/2 -2 < \lfloor 3s/4-\ell/2\rfloor$,
so Phase 1 is not done when bin $b'$ is packed.
Therefore, since $j$ is at most equal to the number of bins in which \newalg\ packed one \Litem, but \opt\ did not,
$j\leq \Lpr+ \LBtpr$.
By Observation~\ref{binsvscounts} and Lemma~\ref{loopinv},
$2\Spr + 2\SBtpr \leq \CS$,
$\Lpr \leq \lfloor (10 \ell -3s)\Spr/(s+2\ell) \rfloor +1$, and
$\LBtpr \leq \SBtpr +1$.
From the algorithm, $\Spr \leq \lfloor (s+2\ell)/16 \rfloor$,
since $\Spr$ increases by at most one each iteration until it reaches
$\lfloor (s+2\ell)/16 \rfloor$ and does not increase thereafter.
Hence $2j -2 \leq 2(\Lpr + \LBtpr -1) \leq 2(1+ (10\ell -3s)\Spr/(s+2\ell) + 1+ \SBtpr -1)
= 2(\Spr + \SBtpr + (8\ell - 4s)\Spr/(s+2\ell)+1)
\leq \CS + 2(8\ell - 4s)/16 + 2
< 3s/4 - \ell/2 -2 + \ell - s/2 + 2 = (s+2\ell)/4$.
Hence,
\begin{eqnarray*}
\frac{\newalg(\sigma,s,\ell)-M}{\OPT(\sigma,s,\ell)} & \leq &
1 + \frac{2j-2}{s+2\ell}\\
& \leq & 1 + \frac{(s+2\ell)/4}{s+2\ell} = \frac{5}{4}.
\end{eqnarray*}

Now, suppose that $6\ell/5 >s$.
If Phase 1 is done, then $\CS \geq \lfloor s/3 \rfloor > s/3 - 1$,
so $2j=s-\CS-i\leq s - s/3 +1 - i \leq 2s/3$.
Otherwise, $\CS < \lfloor s/3 \rfloor$.
Then, since $\CS \geq 2\Spr$, it follows from Lemma \ref{loopinv}
that $j\leq \Lpr \leq 2\Spr + 1 \leq \CS + 1 \leq \lfloor s/3 \rfloor \leq s/3$. 
Hence,
\begin{eqnarray*}
\frac{\newalg(\sigma,s,\ell)-M}{\OPT(\sigma,s,\ell)} & \leq &
1 + \frac{2j}{s+2\ell}\\
& \leq & 1 + \frac{2s/3}{s+2\ell}\\
& < & 1 + \frac{2}{3+6\ell/s}.
\end{eqnarray*}

\paragraph{Case 4: \newalg runs out of \Ss before \OPT  runs out of anything}
Suppose that, in this instance,
\newalg packs its last item of size $S$ in bin $b'$,
but, after  bin $b'$,  \OPT has $i>0$ items of size $S$
and $j>0$ items of size $L$ that are unpacked.
Then,
\newalg has $\ell -\CL$ unpacked items of size $L$,
and,  by Lemma \ref{lemma-reasonable},
$2(\ell -\CL) = i + 2j$.
Note that, in this case, $\ell \geq i/2 + j \geq \frac{1}{2} + 1$, so $\ell \geq 2$.

Since \newalg is reasonable, Lemma~\ref{AlgNoSmall} implies that
\begin{eqnarray*}
\newalg(\sigma,s,\ell) & = & X+R(\sigma',0,\ell -\CL) \\
& \leq & X + \OPT(\sigma',i,j)+\min\{ 0,j-i\} S  +(\ell-\CL-j)M \\
& = & \OPT(\sigma,s,\ell)+\min\{ 0,j-i\} S  +(i/2)M.
\end{eqnarray*}
By Proposition~\ref{optspacelower},
$\OPT(\sigma,s,\ell) \geq sS+\ell L\geq sL/2+\ell L$.
Thus,
\begin{eqnarray*}
\frac{\newalg(\sigma,s,\ell)}{\OPT(\sigma,s,\ell)}
&  \leq & 1 + \frac{(i/2)M+\min\{0,j-i\}S}{sL/2+\ell L}\\
& = & 1 + \frac{i(2L-1) + \min\{0,j-i\}(L+1)}{sL+2\ell L}\\
& = & 1 + \frac{ L(i + \min\{i,j\}) - i + \min\{0,j-i\}}{L(s+ 2\ell)}\\
& \leq & 1+\frac{i + \min\{i,j\}}{s+2\ell}.
\end{eqnarray*}

First suppose $\ell \leq s/2$.  If Phase 1 has not ended,
then the only bins where \newalg had a choice and packed two more \Ss than \OPT did
are the bins contributing to the value of $\SB$, so
$i \leq 2\Spr$.  By Lemma \ref{loopinv},
$\Spr \leq \Lpr \leq \CL$.

If $j\geq i$, then $2(\ell-\CL)=i+2j\geq 3i$,
so $i\leq 2(\ell-\CL)/3$.
In this case, if $\CL \geq \ell/4$, then
$i\leq 2(\ell-\CL)/3 \leq  \ell/2$.
If not, since $\ell \geq 2$,
$\CL<\ell/4 \leq (\ell -1)/2 \leq \lfloor \ell/2 \rfloor$, so Phase 1 has not ended
and $i \leq 2\Spr \leq 2\Lpr \leq 2\CL < \ell/2$.
Therefore, 
$$\frac{\newalg(\sigma,s,\ell)}{\OPT(\sigma,s,\ell)}  \leq 
1 + \frac{2i}{s+2\ell} \leq 
1 + \frac{\ell}{s+2\ell}
= 1 + \frac{1}{2+s/\ell}.$$

Otherwise,
$j<i$.
If $\CL \geq \lfloor \ell/2 \rfloor$,
then $i+j=2(\ell -\CL)-j\leq 2(\ell -\CL)-1 
\leq 2\lceil \ell/2 \rceil -1 \leq \ell$.
If not, $\CL < \lfloor \ell/2 \rfloor$, so Phase 1 has not ended,
$i \leq 2\Spr \leq 2\Lpr \leq 2\CL$, and $i+j =  \ell - \CL +i/2 \leq \ell$.
Hence
$$\frac{\newalg(\sigma,s,\ell)}{\OPT(\sigma,s,\ell)} \leq 
1 + \frac{i+j}{s+2\ell}
\leq 1 + \frac{\ell}{s+2\ell}
\leq 1 + \frac{1}{2+s/\ell}.$$

So, suppose that $\ell > s/2$.
Then $i \leq 2(\Spr + \SBtpr)$, because
the only bins where \newalg had a choice and packed two more \Ss than \OPT did
are the bins contributing to the value of $\SB$ or $\SBt$.
Since \newalg has run out of \Ss, Phase 1 has ended.

Now suppose that $s/2 < \ell \leq 5s/6$.
If $\Spr + \SBtpr \leq (s+2\ell)/16$,
then $i + \min\{i,j\} \leq 2i \leq 4(\Spr + \SBtpr) \leq  (s+2\ell)/4$, so
$$\frac{\newalg(\sigma,s,\ell)}{\OPT(\sigma,s,\ell)}  \leq 
1 + \frac{s+2\ell}{4(s+2\ell)} = \frac{5}{4}.$$

Otherwise, $\Spr  = \lfloor (s+2\ell)/16\rfloor$.
By Lemma~\ref{loopinv}, $\Lpr \geq \lfloor c(\Spr -1) \rfloor + 1$,
where $1 < c = \frac{10\ell -3s}{s+2\ell} \leq 2$ and
$\LBtpr \geq \SBtpr$. Since $\Lpr + \LBtpr \leq \CL$, it follows that
\begin{eqnarray*}
i + \min\{i,j\} \leq i+ j & = & \frac{i}{2} + \ell - \CL\\
& \leq & \Spr + \SBtpr + \ell - \Lpr - \LBtpr\\
& \leq & \ell + \Spr - \lfloor c(\Spr -1) \rfloor - 1\\
& < & \ell + \Spr - c(\Spr -1)\\
& = & \ell + (1-c) \lfloor (s+2\ell)/16\rfloor + c\\
& \leq & \ell + (1-c) ((s+2\ell)/16  - 1) + c\\
& = & \ell + \left( 1-\frac{10\ell - 3s}{s + 2\ell}\right )\frac{s+2\ell}{16} + 2c -1\\
& = & (4s + 8\ell)/16 + 2c-1 \mbox{ and}\\
\frac{\newalg(\sigma,s,\ell)-3}{\OPT(\sigma,s,\ell)}  & \leq &
1 + \frac{ (4s + 8\ell)/16}{s+2\ell}
= \frac{5}{4}.
\end{eqnarray*}

Finally, suppose that $5s/6 < \ell$.
When Phase 1 finished,
$2 \SB \leq \SC \leq \lfloor s/3 \rfloor+ 2$.
Thereafter, $\SB$ does not change. Hence,
$i\leq 2\Spr \leq s/3 + 2$ and
$$\frac{\newalg(\sigma,s,\ell)-4}{\OPT(\sigma,s,\ell)}  \leq 
1 + \frac{2i -4}{s+2\ell}
\leq  1 + \frac{2s/3}{s+2\ell}
= 1 + \frac{2}{3+6\ell/s}.$$
\end{proof}

\section{A Relaxed Version of Grid Scheduling}
\label{sec:relax}
In the definition of the
Grid Scheduling problem, we restricted
the online algorithms so  that
after each bin is packed, no unpacked item fits into its remaining space.
In this section, we consider a relaxed version of Grid Scheduling without this restriction.
Specifically, in this version, a packing can be feasible even if there is a bin which has enough empty space
to contain one of the items packed in a later bin.
We show that the restriction does not make the problem easier.
This may be useful for narrowing the gap between the upper
and lower bounds for the Grid Scheduling problem.
The results below imply that the lower bounds and the properties of optimal packings, which were
proved in Sections \ref{sec:lower} and \ref{sec:properties}, also hold for 
the relaxed version of Restricted Grid Scheduling.

We begin with a lemma that concerns packings of subsets of the items.
If a bin, $b$, in a  packing contains only one item, $i$,
and there is at least one item in a later bin which is no larger than $b$,
removing the item $i$ from the sequence 
might require
substantial repacking to produce a 
feasible
packing.
However,
this lemma shows that there is a feasible repacking that uses 
a (not necessarily proper) subset of the bins used in the original packing.

\begin{lemma}
For any  packing $p$ of a finite set of items $I$ into a sequence of bins $\sigma$ and any item $i \in I$,  there is a  packing $p'$ of $I - \{i\}$ into $\sigma$ that uses a subset of  the bins used by $p$.
If $p$ was produced by an online algorithm, then there exists
an online algorithm that produces $p'$.
\label{validsubset}
\end{lemma}

\begin{proof}
If $|I|=1 $, the claim holds since, after removing one item, no bins are used.
Let $p$ be any packing of a set $I$, with at least two items, into a sequence of bins $\sigma = \langle b_1, b_2, \ldots, b_m\rangle$.
Assume the claim is true for any smaller set of items.

Let $b_k$ be the bin in which $p$ packs item $i$ and let
$J \subseteq I - \{i\}$
be the set of items $p$ packs into later bins.
If $p$ packs more than one item in  bin $b_k$ or all items in $J$ are larger than size($b_k$),
then the packing $p'$ of $I-\{i\}$ that packs every item into the same bins as $p$ does is feasible.

Otherwise, there is an item $j \in J$ that is no larger than $b_k$.
Let $q$ be the  packing of $J$ into $\sigma' = \langle b_{k+1}, \ldots, b_m\rangle$
that packs each item of $J$ into the same bin as $p$ does.
Since the claim holds for any smaller set of items, there is a  packing $q'$ of $J-\{j\}$ that uses a subset of 
the bins used by $q$.
Let $p'$ be the packing of $I-\{i\}$ that
packs the items of
$I-(J \cup \{i\})$
into bins  $b_1,\ldots,b_{k-1}$ the same  as $p$ does,
packs item $j$ into bin $b_k$, and
packs the items of $J - \{j\}$  into bins  $b_{k+1},\ldots,b_m$ the same  as $q'$ does.

Now suppose that the packing $p$ was produced online by an
algorithm \alg. To produce $p'$, an online algorithm $\algprime$, given
the set of items $I-\{i\}$, could simulate \alg on the set of items $I-\{i\}$ together with an imaginary copy of item $i$, until it is about to pack $i$ into bin $b_k$.
If \alg will pack more than one item into $b_k$ or the items that \alg will pack into later bins are all larger than $size(b_k)$, then $\algprime$ continues to simulate \alg
as if it had packed item $i$ into bin $b_k$.
Otherwise, let $j$ be an unpacked item that is no larger than $b_k$
and let $J'$ be the set of the other unpacked items.
Since \alg produces a feasible packing $q$ of $J' \cup \{j\}$ into $\sigma'$,
the induction hypothesis implies that there is 
an online algorithm $\algprimeprime$ that produces the feasible packing $q'$.
Then $\algprime$ packs item $j$ into bin $b_k$ and simulates $\algprimeprime$ on the set of remaining items $J'$.

Hence, the claim is true for $I$. By induction, the claim is true for all finite sets of items $I$.
\end{proof}

A bin in a 
packing
is {\em wasteful} 
if its empty space is at least as large as some unpacked item
or some item packed in a bin that occurs later in the sequence.
In other words, when this bin was packed, there was an extra
item that could have been added to it, but was not.
A
packing is {\em thrifty} if it contains no wasteful bins.
The original version of the problem has the restriction that all feasible packings must be thrifty.

The following result shows 
that we may ignore algorithms which do not produce thrifty packings.

\begin{lemma}
\label{thrift-lem}
For any packing $p$ of a finite set of items $I$  into a sequence of bins $\sigma$,
there is a thrifty packing $p'$ of $I$ into $\sigma$ using a subset of the bins used by $p$.
If $p$ was produced by an online algorithm, then there exists
an online algorithm that produces $p'$.
\end{lemma}

\begin{proof}
Any packing of a set with at most one item is also thrifty.
Let $I$ be a finite set with at least two items and
assume the claim is true for any smaller set of items.

Let $p$ be any packing of $I$ into a sequence of bins $\sigma = \langle b_1, b_2, \ldots, b_m\rangle$. 
Suppose $p$ is not thrifty. Then it has at least one wasteful bin.
Let $b_k$ be the first wasteful bin in $p$, let $e$ be the amount of empty space in $b_k$,  and let 
$J \subseteq I$
be the set of items $p$ packs into later bins.
Since $p$ is feasible, $b_k$ is nonempty, so
$J \subsetneq I$.
Since $b_k$ is wasteful, there is some item $i \in J$ of size at most $e$.

Let $q$ be the packing of $J$ into $\sigma' = \langle b_{k+1}, \ldots, b_m\rangle$
that packs each item of $J$ into the same bin as $p$ does.
By Lemma \ref{validsubset}, there is a packing $q'$ of $J-\{i\}$ into 
 $\sigma'$ that uses a subset of the bins used by $q$.
Let  $b$ be a bin of size $e$, let $\sigma'' = \langle b, b_{k+1}, \ldots, b_m\rangle$,
and let $q''$ be the packing of $J$ into $\sigma''$ that packs item $i$ into bin $b$
and packs each item of $J-\{i\}$ into the same bin that $q'$ does.

Since $J \subsetneq I$, it follows by the induction hypothesis that
there is a thrifty packing $p''$ of $J$ into $\sigma''$ using a subset of
the bins used by $q''$.
Note that $p''$ is feasible and $i\in J$, so $p''$ uses bin $b$.
Moreover, every item $p''$ packs in later bins is smaller than the empty space in $b$.

Let $p'$ be the packing of $I$ into $\sigma$
that packs each item in $I-J$ into the same bin as $p$,
adds to bin $b_k$ the items packed by $p''$ into bin $b$, and
packs bins $b_{k+1}, \ldots, b_m$ as $p''$ packed them.
Since the empty space $p'$ leaves in bin $b_k$ is equal to the empty
space $p''$ leaves in bin $b$, bin $b_k$ is not wasteful in $p'$.
Bins $b_1,\ldots,b_{k-1}$ are not wasteful in $p'$, since they are not wasteful in $p$.
Bins $b_{k+1}, \ldots, b_m$ are not wasteful in $p'$, since they are not wasteful in $p''$.
Therefore $p'$ is a thrifty packing that uses a subset of the bins used by $p$.
By induction, the claim is true for all finite sets of items.

Suppose $p$ was produced by an online algorithm \alg. Since $q$ is also produced by \alg, Lemma  \ref{validsubset} implies that $q'$ is produced by an online algorithm $\algprime$. Then $q''$ can be produced by an online algorithm that packs item $i$ in the first bin that arrives and thereafter simulates algorithm $\algprime$. By the induction hypothesis, it follows that $p''$ can be produced by
an online algorithm $\algprimeprime$.  Finally, an online algorithm that produces $p'$ proceeds by packing each item using \alg until it creates a wasteful bin, then simulates $p''$ with the remaining items, treating the empty space in the wasteful bin as if it were the first bin to arrive, followed by the rest of the sequence.
\end{proof}

The previous lemma shows that restricting attention to thrifty packings is also sufficient for 
the lower bound.
In other words, proving a lower bound on the competitive ratio for the class of thrifty packings
then immediately implies the same lower bound for all packings.
It eliminates the possibility that there is an online packing that is not thrifty and more 
efficient than any online thrifty packing. The following corollary shows that
there is also an optimal (offline) packing which is thrifty.

\begin{corollary}
\label{thrift-cor}
For any optimal packing of a set of items into a sequence of bins, there is an optimal thrifty packing
of those items into that sequence using the same set of bins.
\end{corollary}

\begin{proof}
Let $p$ be an optimal packing of a set of items $I$ into a sequence of bins $\sigma$.
Since $p$ is feasible,
Lemma \ref{thrift-lem} implies that there is a thrifty packing packing $q$
of $I$ into $\sigma$ using a subset of the bins used by $p$.
Thus $cost(q) \leq cost(p)$. But $p$ is optimal, so $cost(q) = cost(p)$.
Hence $q$ is optimal and uses the same set of bins that $p$ uses.
\end{proof}

\section{Conclusions and Open Problems}

Note that the upper bound uses the additive constant in the
definition of the competitive ratio, assuming that
$M$ is a constant. As mentioned earlier, this
bound on $M$ is necessary for the general Grid Scheduling problem
or no algorithm is competitive.
One would like to normalize and assume that $M=1$, since the
lower bound holds with this assumption. The upper bound proof
fails in Case 3 with this assumption because Lemma~\ref{AlgNoLarge}
fails. The proof uses the fact that the empty space in a bin containing
exactly one item of size $S$ is at most $S-1=L-S$ (if that item
is not the last). Thus, it implicitly assumes that there is no bin 
size between
$L$ and $2S$. Without this assumption, the term $(k-s'-\ell'-1)L$
becomes arbitrarily close to $(k-s'-\ell'-1)2S$, increasing the
bounds on the competitive ratio by a factor of $2S/L$.
To obtain a tight upper bound without this assumption,
one would need to change
the algorithm, choosing other values for the ratios defined
by the macros in \newalg. These values are defined using the 
values of $k$
computed in the lower bound proof in Theorem~\ref{thm:lower}. 
We expect that
simply modifying the constants in the macros would give an
algorithm with the same upper bound, but an even more complicated
analysis.

We have shown that varying the proportion of \Ss to \Ls does not
lead to a larger competitive ratio if the maximum bin size is at most $4S-3$.
This may also be
the case
for an arbitrary maximum
bin size, but there are complications.
First,  the invariant $s+2\ell = s'+2\ell'$, which holds until either \opt or the
reasonable packing runs out of one item type, is very heavily used.
For example, in Lemma~\ref{OptNoSmall}
it is used to show that the reasonable packing runs out of large items before \opt does.
With larger bins, the invariant
no longer holds in some cases.
For example, if some bin has size $S\cdot L$, 
one packing could contain $L$ items of size $S$, while the other packing contains $S$ 
items of size $L$.
In addition, it  may be an advantage to have mixed bins which contain
more than one \Sitem and more than one \Litem.
So, when bins
can be large, one needs to consider how to maintain a good
ratio of \Ls to \Ss, as they are packed. We conjecture that
maintaining the ratios specified in \newalg is sufficient.

One could also consider the Grid Scheduling problem for two
bin sizes, $S$ and $L \neq 2S-1$.
For example, when $L$ is a multiple of $S$,
then the algorithm of Example 2,
which always uses as many items of size $L$
as possible, is optimal and, hence, has competitive ratio 1.
It would be interesting to see if changing the ratio of $S$ to
$L$ could improve the lower bound.
We conjecture that it does not.

For the general problem, where there could be more than
two  sizes of items,
we would like to
close the gap between our lower bound
and the upper bound of $\frac{13}{7}$ in \cite{BF10}.

Finally, it would be interesting to consider the competitive ratio
of randomized algorithms for the Grid Scheduling problem
or Restricted Grid Scheduling problem against an oblivious adversary. 

\section{Acknowledgements}
Joan Boyar was supported in part by the Villum Foundation and the
Danish Council for Independent Research,
Natural Sciences (FNU).
Part of this work was carried out while she was visiting
 the University of Waterloo and the University of Toronto.

Faith Ellen was supported in part by the VELUX Foundation
and the Natural Science and Engineering Research Council of Canada
(NSERC). 
Part of this work was carried out while she was visiting
 the University of Southern Denmark.

\end{document}